%% file: Main_Reply.tex
\documentclass[a4paper,amsfonts,twocolumn,superscriptaddress,nofootinbib,longbibliography,accepted=2020-12-18]{quantumarticle}
\pdfoutput=1

\usepackage[numbers,sort&compress]{natbib}

\usepackage{epsfig,amsmath,amssymb,bm,epsf,graphicx,psfrag,bbm}
\usepackage{color, comment, verbatim}
\usepackage{ragged2e}
\usepackage{subfigure}
\usepackage{enumerate}
\usepackage{chngcntr}

\usepackage[numbers,sort&compress]{natbib}
\usepackage[colorlinks,breaklinks]{hyperref}
\usepackage{dsfont}
\usepackage{float}
\usepackage{tabularx}
\usepackage{graphicx}
\usepackage{amsbsy}
\usepackage{amsthm}

\usepackage{dsfont} 

\usepackage{tikz}
\usepackage[toc,page]{appendix}
\usepackage{titlesec}
\usepackage{xcolor}

\makeatletter
\newcommand\org@hypertarget{}
\let\org@hypertarget\hypertarget
\renewcommand\hypertarget[2]{%
  \Hy@raisedlink{\org@hypertarget{#1}{}}#2%
  }
\makeatother
\hypersetup{
	bookmarksnumbered,
	pdfstartview={FitH},
	citecolor={darkgreen},
	linkcolor={darkred},
	urlcolor={darkblue},
	pdfpagemode={UseOutlines}}
\definecolor{darkgreen}{RGB}{50,190,50}
\definecolor{darkblue}{RGB}{0,0,190}
\definecolor{darkred}{RGB}{238,0,0}

\usepackage[framemethod=tikz]{mdframed}
\definecolor{mycolor}{rgb}{0.122, 0.435, 0.698}
\definecolor{mycolor2}{RGB}{112, 48, 160}
\newmdenv[innerlinewidth=0.5pt, roundcorner=4pt,linecolor=mycolor,innerleftmargin=6pt,
innerrightmargin=6pt,innertopmargin=6pt,innerbottommargin=6pt]{mybox}

\usepackage{paracol}
\usepackage{tcolorbox}
\tcbset{colback=mycolor2!5,colframe=mycolor2,
fonttitle=\bfseries, float=htb}
\newtcolorbox[blend into=figures]{boxfigure}[3][]
{ float*=ht,width=\textwidth,lower separated=false, center upper,
title={#2},label= fig:#3,#1}
\newtcolorbox[blend into=figures]{smallboxfigure}[3][]
{float=ht,lower separated=false, blend before title=colon hang,
title={#2}, label= fig:#3 ,#1}
\newtcolorbox{smallbox}[3][]
{float=ht,lower separated=false, blend before title=colon hang,
title={#2}, label= fig:#3 ,#1}
\newtcolorbox[blend into=tables]{smallboxtable}[3][]
{float=tb,lower separated=false, blend before title=colon hang,
title={#2}, label= table:#3 ,#1}
\newtcolorbox[blend into=tables]{bigboxtable}[3][]
{float*=t,lower separated=false, blend before title=colon hang, width = 2\linewidth,
title={#2}, label= table:#3 ,#1}
\newcolumntype{Z}{|>{\centering\arraybackslash}X}
\usepackage{enumerate}
\hypersetup{
	bookmarksnumbered,
	pdfstartview={FitH},
	citecolor={darkgreen},
	linkcolor={darkred},
	urlcolor={darkblue},
	pdfpagemode={UseOutlines}}
\definecolor{darkgreen}{RGB}{50,190,50}
\definecolor{darkblue}{RGB}{0,0,190}
\definecolor{darkred}{RGB}{238,0,0}
\usepackage{soul}

\newcommand{\be}{\begin{equation}}
\newcommand{\ee}{\end{equation}}
\newcommand{\ben}{\begin{equation*}}
\newcommand{\een}{\end{equation*}}
\newcommand{\bea}{\begin{eqnarray}}
\newcommand{\eea}{\end{eqnarray}}
\newcommand{\pri}{^{\prime}}
\newcommand{\tr}{\textnormal{Tr}}

\newcommand{\half}{\mbox{$\textstyle \frac{1}{2}$}}

\newcommand{\ket}[1]{\ensuremath{\left|\right.\!{#1}\!\left.\right\rangle}}

\newcommand{\bra}[1]{\ensuremath{\left\langle\right.\!{#1}\!\left.\right|}}

\newcommand{\ketbra}[2]{\ensuremath{|{#1}\rangle\langle{#2}|}}
\newcommand{\brakket}[3]{\ensuremath{\langle{#1}|{#2}|{#3}\rangle}}
\newcommand{\scpr}[2]{\ensuremath{\left\langle\right.\hspace*{-1pt} #1 \hspace*{-1pt}\left|\right.\hspace*{-1pt} #2 \hspace*{-1pt}\left.\right\rangle}}

\newcommand{\nr}{\ensuremath{\hspace*{0.5pt}}}

\newcommand{\suptiny}[3]{\ensuremath{^{\hspace{#1 pt}\protect\raisebox{#2 pt}{\tiny{$ #3$}}}}}
\newcommand{\SA}{\ensuremath{_{\hspace{-1pt}\protect\raisebox{0pt}{\tiny{$A$}}}}}

\newcommand{\SB}{\ensuremath{_{\hspace{-1pt}\protect\raisebox{0pt}{\tiny{$B$}}}}}

\newcommand{\SAB}{\ensuremath{_{\hspace{-1pt}\protect\raisebox{0pt}{\tiny{$A\hspace*{-0.5pt}B$}}}}}

\DeclareMathOperator{\sinc}{sinc}

\newcommand{\djj}{d\kern-0.4em\char"16\kern-0.1em}

\newtheorem{lemma}{Lemma}

\newcolumntype{s}{>{\hsize=.6\hsize}X}

\newcommand{\nh}[1]{\textcolor{black}{#1}}

\newcommand{\nf}[1]{\textcolor{black}{#1}}

\begin{document}

\title{High-Dimensional Pixel Entanglement: Efficient Generation and Certification}

\author{Natalia Herrera Valencia}
\email[Email address: ]{nah2@hw.ac.uk}
    \affiliation{Institute of Photonics and Quantum Sciences, Heriot-Watt University, Edinburgh, UK}
    \orcid{0000-0003-1468-8953}
    
\author{Vatshal Srivastav}
    \affiliation{Institute of Photonics and Quantum Sciences, Heriot-Watt University, Edinburgh, UK}
    
\author{Matej Pivoluska}
    \affiliation{Institute of Physics, Slovak Academy of Sciences, Bratislava, Slovakia}
    \affiliation{Institute of Computer Science, Masaryk University, Brno, Czech Republic}  
    \orcid{0000-0001-6876-2085}

\author{Marcus Huber}
    \affiliation{Institute for Quantum Optics and Quantum Information - IQOQI Vienna, Austrian Academy of Sciences, Vienna, Austria}
    \affiliation{Vienna Center for Quantum Science and Technology, Atominstitut, TU Wien,  1020 Vienna, Austria}%
    \orcid{0000-0003-1985-4623}

\author{Nicolai Friis}
    \affiliation{Institute for Quantum Optics and Quantum Information - IQOQI Vienna, Austrian Academy of Sciences, Vienna, Austria}
    \orcid{0000-0003-1950-8640}
    
\author{Will McCutcheon}
    \affiliation{Institute of Photonics and Quantum Sciences, Heriot-Watt University, Edinburgh, UK}
    \orcid{0000-0003-0344-6385}
    
\author{Mehul Malik}
    \email[Email address: ]{m.malik@hw.ac.uk}
    \affiliation{Institute of Photonics and Quantum Sciences, Heriot-Watt University, Edinburgh, UK}
    \affiliation{Institute for Quantum Optics and Quantum Information - IQOQI Vienna, Austrian Academy of Sciences, Vienna, Austria}
    \orcid{0000-0001-8395-160X}

\begin{abstract}
Photons offer the potential to carry large amounts of information in
their spectral, spatial, and polarisation degrees of freedom. While state-of-the-art classical communication systems routinely aim to maximize this information-carrying capacity via wavelength and spatial-mode division multiplexing, quantum systems based on multi-mode entanglement usually suffer from low state quality, long measurement times, and limited encoding capacity. At the same time, entanglement certification methods often rely on assumptions that compromise security. Here we show the certification of photonic high-dimensional entanglement in the transverse position-momentum degree-of-freedom with a record quality, measurement speed, and entanglement dimensionality, without making any assumptions about the state or channels. Using a tailored macro-pixel basis, precise spatial-mode measurements, and a modified entanglement witness, we demonstrate state fidelities of up to 94.4\% in a 19-dimensional state-space, entanglement in up to 55 local dimensions, and an entanglement-of-formation of up to 4 ebits. Furthermore, our measurement times show an improvement of more than two orders of magnitude over previous state-of-the-art demonstrations. Our results pave the way for noise-robust quantum networks that saturate the information-carrying capacity of single photons.
\end{abstract}
\maketitle

\section{Introduction}
\label{sec:Intro}
Quantum entanglement plays a pivotal role in the development of quantum technologies, resulting in revolutionary concepts in quantum communication such as superdense coding~\cite{Bennett2002} and device-independent security~\cite{Vazirani:2014hu,Acin:2007db}, as well as enabling fundamental tests of the very nature of reality~\cite{Giustina:2015fc,Shalm:2015wu,Proietti:2019bv}. While many initial demonstrations have relied on entanglement between qubits, recent advances in technology and theory now allow us to fully exploit high-dimensional quantum systems. 
In particular, the large dimensionality 
offered by photonic quantum systems has provided the means for quantum communication with record capacities~\cite{Islam:2017hs,Mirhosseini:2015fy}, noise-resistant entanglement distribution~\cite{Ecker-Huber2019,Zhu:2019tb}, robust loophole-free tests of local realism~\cite{Salavrakos:2017ek,Vertesi:2010bq}, and scalable methods for quantum computation~\cite{Gokhale2019}. 

In order to make full use of the potential of high-dimensional entanglement, it is of key importance to achieve the certification of entanglement with as few measurements as possible. The characterisation of a bipartite state with local dimension $d$ through full state tomography requires $\mathcal{O}(d^4)$ single-outcome projective measurements~\cite{Friis:2019hg,Agnew2011}, making this task extremely impractical in high dimensions. More efficient tools for quantifying high-dimensional entanglement involve entanglement witnesses that use semi-definite programming~\cite{Martin2017}, matrix completion techniques~\cite{Tiranov:2016wa}, or compressed sensing~\cite{Schneeloch:2019uw}. In this context, it is crucial to certify entanglement without compromising the security and validity of the applications by introducing assumptions on the state, e.g., purity of the generated state~\cite{Kues:2017db} or conservation of quantities~\cite{Krenn:2014jy}. Recent 
work has shown how measurements in mutually unbiased bases (MUBs) allow the efficient certification of high-dimensional entanglement~\cite{Erker:2017cb,Bavaresco:2018gw}, without any such assumptions on the underlying quantum state. Here we employ and extend this method to improve upon the quality and speed of high-dimensional entanglement certification.

High-dimensional entanglement has been demonstrated in multiple photonic platforms, with encodings in the orbital angular momentum (OAM)~\cite{Dada:2011dn,Bavaresco:2018gw}, time-frequency~\cite{Martin2017}, path~\cite{Wang:2018gh}, and transverse position-momentum degrees of freedom (DoF)~\cite{Schneeloch:2019uw}. While time-frequency encoding offers the potential of accessing spaces with very large dimensions, the difficulty of measuring coherent superpositions of multiple time-bins hinders the scalability of the technique, and in turn necessitates certification methods that require unwanted assumptions on the reconstruction of the state in question to reach their full potential~\cite{Martin2017}. Path-encoding in integrated photonic circuits offers yet another promising avenue for realising high-dimensional entanglement~\cite{Wang:2018gh}. However, the precise fabrication and control of $\frac{d(d-1)}{2}$ Mach-Zehnder interferometers required for universal operations in $d$ dimensions poses significant practical challenges as the dimension is increased~\cite{Flamini:2017fy}.

Meanwhile, techniques for the creation, manipulation and detection of entanglement in photonic OAM bases have seen rapid progress in recent years~\cite{Krenn2017,Erhard:2018iua}, where devices such as spatial light modulators (SLMs) enable generalized measurements of complex amplitude modes. However, such measurements necessarily suffer from loss~\cite{Arrizon:2007wl} and have limited quality~\cite{Bouchard:2018hr, Qassim:2014fp}, resulting in long measurement times and reduced entanglement quality in large dimensions~\cite{Bavaresco:2018gw}. An alternative choice of basis is the discretised transverse position-momentum, or ``pixel'' basis~\cite{OSullivan:2005vh}, where the lack of efficient single-photon detector arrays necessitates scanning through localised position or momentum modes, or subtracting a large noise background~\cite{Edgar:2012bl}. Such measurements are also subject to extreme loss, and as a result require long measurement times and strong assumptions on detector noise, such as background or accidental count subtraction.

In this work, we report on significant progress towards overcoming the challenges of scalability, speed, and quality in the characterisation of high-dimensional entanglement with a strategy that combines three distinct improvements in the generation and measurement of spatially entangled modes. Working in the discretized transverse position-momentum DoF, we first tailor our spatial-mode basis by adapting it to the characteristics of the two-photon state generated and measured in our experimental setup. Second, we implement a recently developed spatial-mode measurement technique~\cite{Bouchard:2018hr} that ensures precise projective measurements in any mode basis of our choice. Third, we generalise a recently developed entanglement dimensionality witness~\cite{Bavaresco:2018gw} to certify high-dimensional entanglement using any two high-dimensional MUBs. Crucially, this allows us to bypass lossy localised mode measurements in the transverse position or momentum bases. 

The combination of these improvements in basis optimisation, spatial-mode measurement, and certification tools allow us to certify high-dimensional entanglement with a record quality, speed, and dimension, reaching state fidelities up to 98\%, certified entanglement dimensionality up to 55 (in local dimension 97), and an entanglement-of-formation of up to 4 ebits. Below we introduce our theoretical framework and elaborate on our improved techniques.


\section{Theory}

A two-photon state entangled in transverse position-momentum and produced via the process of spontaneous parametric down-conversion (SPDC) is characterised by its \emph{joint-transverse-momentum-amplitude} or JTMA (Fig.~\ref{fig:JTMA}~a), which is well approximated by the function~\cite{Schneeloch:2016ch,Miatto:2012ck}
\begin{align}
\label{eq:JTMAformula}
F(\vec k_s, \vec k_i)= \exp\Bigl\{-\half \frac{|\vec k_s+\vec k_i|^2}{\sigma_P^2}\Bigr\} 
\sinc\Bigl\{\frac{1}{\sigma_S^2}|\vec k_s - \vec k_i|^2\Bigr\},
\end{align}
where $\vec k_s$ and $\vec k_i$ are the transverse momentum components of the signal and idler photons. The parameters $\sigma_P$ and $\sigma_S$ are the widths of the minor and major axis of the JTMA, which are dependent on the pump transverse momentum bandwidth and the crystal characteristics, respectively. In effect, $\sigma_P$ indicates the degree of momentum correlation, while $\sigma_S$ is a measure of the number of correlated modes in the state.

Experimental studies of spatial-mode entanglement usually involve projective measurements made with a combination of a hologram and single-mode fiber (SMF), which together act as spatial-mode filter. The addition of collection optics for coupling to an SMF results in the \textit{collected} bi-photon JTMA of
\be
\label{eq:collectedJTMA}
G(\vec k_s, \vec k_i)=
\exp\Bigl\{-\half \frac{|\vec k_s|^2}{\sigma_C^2}\Bigr\}
\exp\Bigl\{-\half \frac{|\vec k_i|^2}{\sigma_C^2}\Bigr\}
F(\vec k_s, \vec k_i),
\ee
where $\sigma_C$ denotes the collection bandwidth, that is, the size of the back-projected detection mode. The effect of the collection on the correlations is illustrated in Fig.~\ref{fig:JTMA}, where the Gaussian functions attributed to the SMF modes decrease the probability of detecting higher-order modes associated with the edges of the JTMA. Finally, we include the hologram functions $\Phi_s(\vec k_s)$ and $\Phi_i(\vec k_i)$ used for performing projective measurements on the signal and idler photons, respectively, which leads to the two-photon coincidence probability
\be
\Pr(\Phi_s,\Phi_i)=\biggl| \int d^2\vec k_s \int d^2 \vec k_i \Phi_s(\vec k_s) \Phi_i(\vec k_i) G(\vec k_s, \vec k_i)\biggr|^2.
\ee

Maximizing the quality and dimensionality of an experimentally generated entangled state involves a complex interplay of optimizing the generation and measurement parameters introduced above. Recent work in this direction has focused on pump shaping~\cite{Kovlakov:2018gz,Liu2018} and entanglement witnesses that adapt to the non-maximally entangled nature of the state~\cite{Bavaresco:2018gw}. Below, we introduce a new approach to engineering high-quality spatial-mode entanglement based on three distinct improvements: tailoring the spatial-mode measurement basis, precise two-photon measurements via intensity-flattening, and a modified witness for high-dimensional entanglement.


\begin{figure*}[ht!]
\centering\includegraphics[width=0.7\textwidth]{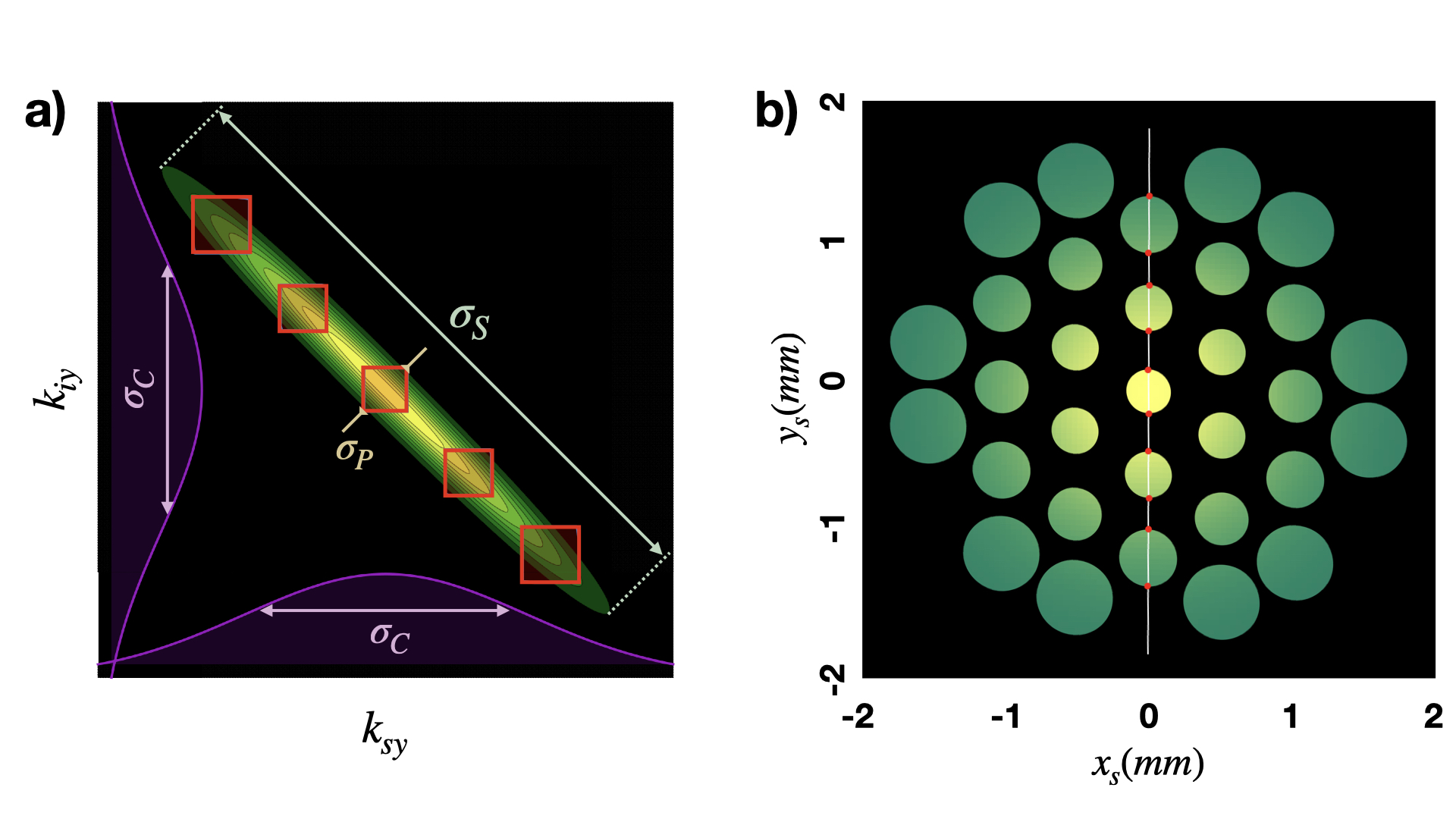}
\vspace*{-4mm}
\caption{\textbf{Joint-transverse-momentum-amplitude (JTMA) and pixel basis tailoring.} a) A contour plot (green/yellow) depicting the absolute value of a 2-d slice $F\left( (0,k_{sy}),(0,k_{iy})\right)$ of the JTMA corresponding to $k_{sx}=k_{ix}=0$. The Gaussian distributions (purple) on the $k_{sy}$ and $k_{iy}$ axes are the collection mode envelopes $\exp\{-\half \left(\frac{k_{sy}}{\sigma_C}\right)^2\}$. The square regions indicate values of $k_{sy}$ and $k_{iy}$ used for generating an optimised pixel basis mask. As such, integrals over the square regions are closely related to the coincidence rate for projections on to the given pair of pixels (up to the additional 2-d integral over $k_{sx}$ and $k_{ix}$), demonstrating the necessity to increase the pixel radii for those positioned over the less intense regions of the JTMA. b) An example 31-dimensional tailored pixel mask. The color function displays the marginal JTMA intensity of the signal photon corresponding to $\int d^2 \vec k_i F(\vec k_s, \vec k_i) $.  The line at $x_s=0$ intersects 5 pixels having 3 unique radii, and its intersections with their boundaries are marked by red dots. These points are mapped through the optical system onto the corresponding boundaries of the regions in momentum space at the crystal shown in the preceding figure.}
\label{fig:JTMA}
\end{figure*}

\subsection{Pixel basis design}\label{sec:pixel basis design}

In our experiment, we choose to work in the discretized transverse-momentum, or macro-pixel basis. This basis provides several advantages over other spatial mode bases. First, projective measurements in bases that are unbiased with respect to the pixel basis require phase-only measurements. In contrast to the lossy amplitude and phase measurements required for Laguerre-Gauss modes~\cite{Qassim:2014fp, Bouchard:2018hr} and their superpositions, measurements with phase-only modulation are lossless in theory and produce a count rate that is independent of the specifics of the chosen basis. As a result, measurements in such bases maximize photon flux, are resistant to detector noise, and allow us to minimize the total number of measurements required. Second, as the distribution of macro-pixels is determined by circle packing formulas, this basis is compatible with state-of-the-art quantum communication technologies based on multi-core fibres~\cite{Carine:2020te,DaLio:2020jp,Lee:2019,Gomez:2020} and was recently employed for high-dimensional entanglement transport through a commercial multi-mode fibre~\cite{Valencia2019} as well as a test of genuine high-dimensional steering~\cite{Designolle2020}. Third and most significantly, informed by knowledge of the JTMA, we can optimise this basis to approach a maximally entangled state of the form $\ket{\Phi}=\frac{1}{\sqrt{d}}\sum_{m}\ket{mm}$, where $\ket{m}$ stands for one photon in mode m and none in the others. This in turn enables us to use powerful theoretical techniques that rely on mutually unbiased bases for entanglement certification. In contrast with Procrustean filtering techniques that achieve a similar goal by adding mode-dependent loss~\cite{Vaziri:2003cx,Dada:2011dn}, this optimisation can be done in a relatively lossless manner by tailoring the size and spacing of individual pixel modes.

As illustrated in Fig.~\ref{fig:JTMA}, the JTMA indicates that outer macro-pixel modes exhibit the strongest correlation albeit with lower amplitudes, while the inner pixel modes show weaker correlation with higher amplitudes. In order to obtain the highest fidelity to the maximally entangled state, the spacing and size of pixels can be optimised to minimise cross-talk arising from non-vanishing pump transverse bandwidth $\sigma_P$, while simultaneously equalizing pixel probability amplitudes and maximizing photon count rates. This proceeds as follows: the diameter of the circular region containing the macro-pixel basis is determined by the width $\sigma_S$ of the JTMA major axis, such that the outermost pixel modes for a chosen dimension have sufficient amplitude. The maximum coincidence rate for the outermost pixels is obtained by giving them the maximum radius allowed for a given dimension and the chosen spacing. Taking into account the decreasing correlation strength, we then proceed to choose the radius for the inner pixels such that the photon count rate is equal for all pixels, thereby approximating a state with equal Schmidt coefficients extremely well.


\begin{figure*}[ht!]
\centering\includegraphics[width=\textwidth]{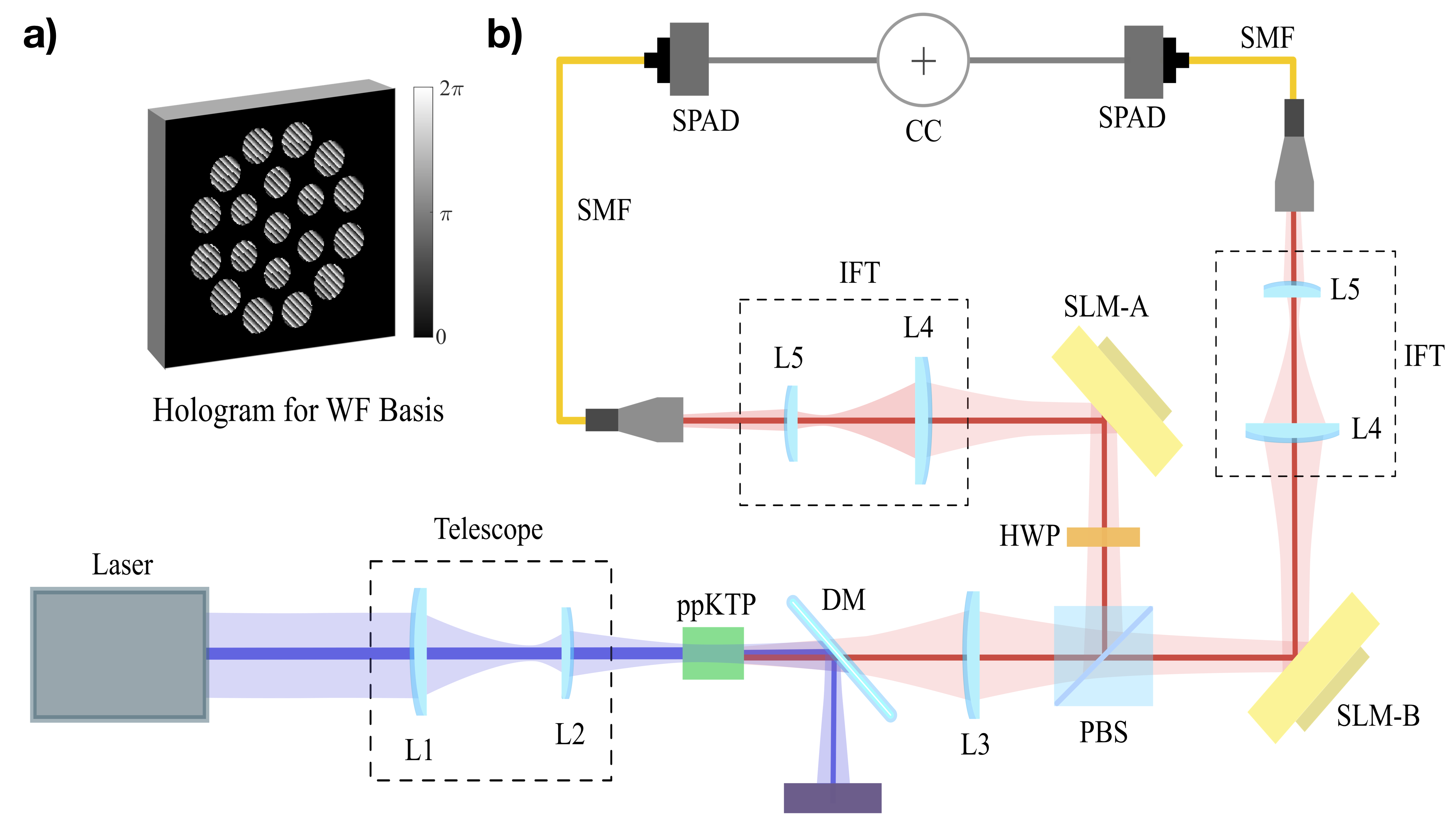}
\vspace*{-6mm}
\caption{\textbf{Experimental Setup.} a) An example computer-generated hologram used for projective measurements in a \nh{19-dimensional Wootters-Fields (WF) basis state (Eq.~(\ref{eq:WFbasis}))}. b) A grating-stabilized UV laser (405~nm) is shaped by a telescope system of lenses and used for generating a pair of infrared photons (810~nm) entangled in their transverse position-momentum via Type-II spontaneous-parametric-down-conversion (SPDC) in a non-linear ppKTP crystal. After removing the UV pump with a dichroic mirror (DM), the photons are separated with a polarising-beam-splitter (PBS) and made incident on two phase-only spatial light modulators (SLM) at an angle of $5^{\circ}$ (the $45^{\circ}$ angle shown is solely for the purpose of the illustration). Precise projective measurements in the pixel basis and any of its mutually unbiased bases are performed with the combination of SLMs, intensity-flattening telescopes (IFT), and single-mode-fibers. The filtered photons are detected by single-photon-avalanche-detectors (SPAD) connected to a coincidence counting logic (CC) that records time coincident events with a coincidence window of 0.2~ns.}
\label{fig:ExpSetup}
\end{figure*}

\subsection{Intensity-flattening telescopes}

As shown in Fig.~\ref{fig:ExpSetup}, the detection of our entangled state depends on a spatial-mode filtering scheme composed of holographic mode projectors implemented on spatial light modulators (SLMs) and single-mode fibers (SMFs). As Eq.~(\ref{eq:collectedJTMA}) indicates, the collection bandwidth of the entangled state is limited by the SMF Gaussian mode width $\sigma_C$ that depends on the specific characteristics of the fiber and coupling optics used. As a result, higher order modes are relatively suppressed, which has an especially adverse effect when measuring coherent high-dimensional superpositions of spatial modes. In recent work, we demonstrated an ``intensity-flattening" technique that dramatically improves the quality of general mode-projective measurements at the expense of adding loss~\cite{Bouchard:2018hr}.

This technique can be extended to the two-photon case, as depicted by the illustration of the collection modes superposed on the JTMA (represented by the purple curves in Fig.~\ref{fig:JTMA}~a). If the Gaussian back-propagated modes are made wider, higher order modes associated with the edges of the JTMA will couple more efficiently to the SMF, while lower order modes are relatively suppressed. This effectively allows us to make the approximation $G(\vec k_s, \vec k_i) \approx F(\vec k_s, \vec k_i)$  in Eq.~(\ref{eq:collectedJTMA}), equating the collected JTMA with the generated one. We implement this technique by using intensity-flattening telescopes (IFTs) that afocally magnify the back-propagated collection modes to a size that optimally increases the two-photon collection bandwidth while keeping loss at a tolerable level. While our previous work demonstrated the IFT technique with classical light from a laser, this work extends it to an entangled state for the first time, demonstrating a marked improvement in reconstructed two-photon state fidelity~\cite{Bavaresco:2018gw}. For additional details on the experimental setup we refer to Appendix~\ref{sec:appendix A.I}.


\subsection{High-dimensional entanglement witness}

Our entanglement certification technique builds on an entanglement-dimensionality witness that uses measurements in the chosen standard basis, complemented by measurements in a basis that is unbiased with respect to it, in order to efficiently certify high-dimensional entanglement~\cite{Bavaresco:2018gw}. Here, we briefly introduce this witness, identify its limitations when applied to experiment, and discuss our modifications that significantly improve its utility. 

The local Schmidt basis of a position-momentum entangled state, Eq.~(\ref{eq:JTMAformula}), which we designate as the \emph{standard basis}, normally consists of localised transverse spatial (or momentum) modes. When performing single-outcome projective measurements (as described above) in such a basis, the detector count rates are restricted by the limited projection onto the collection mode (a pixel). Therefore, measurements in such a spatially localised basis result in significantly fewer counts (per unit time) than measurements in bases that are unbiased\footnote{For brevity, we will refer to such bases simply as `unbiased' from here on, which is to be understood as being in reference to a chosen standard basis.} with respect to the spatially localized basis. 
This is the case because the vectors in the unbiased basis correspond to equally weighted superpositions of the spatially localised modes. Measurements in a spatially delocalised unbiased basis can thus collect photons incident on any of the pixels.  

As a result, `unbiased' measurements take $\mathcal{O}(d)$ less time individually, or $\mathcal{O}(d^2)$ less time in the bipartite case, than localised standard basis single-outcome projective measurements for a given photon flux. As we demonstrate here, one may construct an efficient witness for certifying high-dimensional entanglement purely from measurements in two (or more) spatially delocalised bases that are mutually unbiased with respect to each other (and with respect to the original standard basis), without resorting to measurements in the standard basis at all. This construction is extremely beneficial and permits us to lower the measurement times for certifying high-dimensional entanglement significantly.

The method described in Ref.~\cite{Bavaresco:2018gw} allows one to certify high-dimensional entanglement by estimating a lower bound on the fidelity $F(\rho,\Phi)$ of a measured state $\rho$ with a chosen target state $\ket{\Phi}=\frac{1}{\sqrt{d}}\sum_{m}\ket{mm}$. By performing two-photon measurements in the standard basis $\{\ket{mn}\}$ and a second basis $\{\ket{\tilde{m}_{k}\tilde{n}_{k}^*}\}$ that is unbiased with respect to the standard basis, one can obtain a lower bound to the fidelity $F(\rho,\Phi)$. As explained in Ref.~\cite{Bavaresco:2018gw}, the fidelity can in turn be used to bound the Schmidt number of $\rho$ from below.
That is, the certified entanglement dimensionality is the maximal $d_{\mathrm{ent}}$ such that $(d_{\mathrm{ent}}-1)/d \geq F(\rho,\Phi)$. The unbiased bases can be constructed in a standard manner by following the prescription by Wootters and Fields~\cite{Wootters1989}, i.e.,
\begin{align}
    \ket{\tilde{j}_{k}} &= \tfrac{1}{\sqrt{d}} \sum_{m=0}^{d-1} \omega^{jm+km^{2}} \ket{m},
    \label{eq:WFbasis}
\end{align}
where $\omega=\exp(\tfrac{2\pi \nr i}{d})$ is the principal complex d$-th$ root of unity, $k\in\{0,\dots,d-1\}$ labels the chosen basis, and $j\in\{0,\dots,d-1\}$ labels the basis elements. We refer to these bases as Wootters-Fields (WF) bases. Notice that when $d$ is an odd prime, which we will assume from now on (see Appendix~\ref{sec:appendix A.II}), the set of $d$ bases $\{\ket{\tilde{j}_k}\}_j$ together with the standard basis $\{\ket{m}\}_m$ forms a complete set of $d+1$ mutually unbiased bases (MUBs) with the property that the overlap between any two basis states chosen from any two different bases in the set have the same magnitude.

In order to certify entanglement without using the standard basis, we use a property of \nf{the} maximally entangled \nf{state $\ket{\Phi}=\frac{1}{\sqrt{d}}\sum_{m}\ket{mm}$}. Such states are invariant under transformations $ (U\otimes U^*)$ for any unitary operator $U$, and thus have the same form in any WF basis, allowing us to express our target state as $\ket{\Phi} = \frac{1}{\sqrt{d}}\sum_{m}\ket{\tilde{m}_k\tilde{m}^*_k}$. The fidelity of our experimental state to this target state $F(\rho,\Phi)$ can then be expressed in terms of $k$-th WF basis and be split into two contributions, $F(\rho,\Phi)=F_{1}(\rho,\Phi)+F_{2}(\rho,\Phi)$, where
\begin{subequations}
\begin{align}
    F_{1}(\rho,\Phi)    &:=\tfrac{1}{d}\sum_{m}\bra{\tilde{m}_k\tilde{m}^*_k}\rho\ket{\tilde{m}_k\tilde{m}^*_k},\\[1mm]
    F_2(\rho,\Phi)   &:=\tfrac{1}{d}\sum_{m\neq n} \bra{\tilde{m}_k\tilde{m}_k^*}\rho\ket{\tilde{n}_k\tilde{n}^*_k}.
    \label{eq:FidelityinMUB}
\end{align}
\end{subequations}
While the first term $F_{1}(\rho,\Phi) $ can be calculated directly from measurements in $k$-th WF basis $\{\ket{\tilde{m}_{k}\tilde{n}_{k}^*}\}$, we show in Appendix~\ref{sec:some name} how one can determine a lower bound for $F_2(\rho,\Phi)$ with measurements in a second WF basis (labeled by $k\pri\neq k$ in the construction above) $\{\ket{\tilde{m}_{k\pri}\tilde{n}_{k\pri}^*}\}$, allowing us to calculate a lower bound for the fidelity to the maximally entangled state given by
\begin{align}
    & F(\rho,\Phi)    \ge\, \tilde{F}(\rho,\Phi) =\,
    \tfrac{1}{d}\sum_{m}\bra{\tilde{m}_{k}\tilde{m}_{k}^{*}}\rho\ket{\tilde{m}_{k}\tilde{m}_{k}^{*}}\,\nonumber\\[0.5mm]
    &\hspace*{4mm}
    +
    \sum\limits_{m}\bra{\tilde{m}_{k\pri}\tilde{m}_{k\pri}^{*}}\rho\ket{\tilde{m}_{k\pri}\tilde{m}_{k\pri}^{*}}
    -\,\tfrac{1}{d}\label{eq:fid bound two WF bases_main}\\[0.5mm]
    &-\hspace*{-2.5mm}
    \sum\limits_{\substack{m\neq m\pri\!\\  m\neq n \\ n\neq n\pri\!\\ n\pri\neq m\pri\\}}\hspace*{-1.5mm}
    \tilde{\gamma}_{mnm\pri n\pri}\,
    \sqrt{
     \bra{\tilde{m}_{k}\pri \tilde{n}_{k}^{\prime\nr *}}\rho\ket{\tilde{m}_{k}\pri \tilde{n}_{k}^{\prime\nr *}}
     \bra{\tilde{m}_{k}\tilde{n}_{k}^{*}}\rho\ket{\tilde{m}_{k}\tilde{n}_{k}^{*}}
    }\, ,\nonumber
\end{align}
where the term $\tilde{\gamma}_{mnm\pri n\pri}$ vanishes whenever $(m-m\pri-n+n\pri)\!\!\mod(d)\neq0$, and is equal to $\frac{1}{d}$ otherwise.

This result shows that measurements in any two WF bases can be used to lower bound the fidelity to a target state, and hence to certify the entanglement dimensionality. In particular, this allows us to bypass the use of measurements in the standard basis used in~\cite{Bavaresco:2018gw}, while providing significant flexibility in the measurement settings. 
By construction, the witness works best if the chosen target state is close to the state $\rho$ whose entanglement is being certified. Nevertheless, the witness is valid for any choice of target state, and thus requires no assumptions about the state $\rho$. \nh{We do make certain assumptions about our local measurement devices, which means that our witness is not fully device-independent. Below, we discuss these assumptions on the state and devices in the context of our experiment.}

While full device-independence is challenging to achieve, almost beyond practicability with current technology, it is usually also more than what is required for practical security. Assumptions on states and devices can enter in different form, including assumptions on certain properties of the state (purity, conservation of energy or momentum), to assuming certain device properties (perfect measurements, background/dark count subtraction). Our only assumption is that our measurements work in a $d$-dimensional Hilbert space and that the eigenstates corresponding to measurement outcomes in different bases are mutually unbiased with a specific phase relation. This assumption can and has been tested by preparing eigenstates of one basis and measuring in another basis, to estimate if all possible results in the latter basis are equally distributed \cite{Bouchard:2018hr}. No further assumptions on the states are needed in our setting, and we do not perform any background subtraction. All assumptions that we make can thus be tested experimentally.
We believe that this presents a reasonable compromise between fully device-independent security and practicability of implementation. Moreover, while entanglement witnesses based on MUBs can in principle be constructed in a device-independent way~\cite{TavakoliFarkasRossetBancalKaniewski2019}, device-independent certification of the Schmidt number is generally not possible~\cite{HirschHuber2020}.


\section{Results}

We implement the above improvements in basis optimisation, spatial-mode measurements, and entanglement certification in a two-photon entanglement experiment. Figure~\ref{fig:ExpSetup}~a) depicts a tailored diffractive hologram used for performing spatial-mode projective measurements in a 19-dimensional \nh{WF basis. Notice that contrary to the holograms used for measuring in the standard basis, where projections on each state are made by  ``switching on'' a single macro-pixel at a time, the WF basis states require us to display all macro-pixels at once, each with an appropriate phase selected according to Eq. (\ref{eq:WFbasis}).}
 
The pixel orientations are determined by circle packing formulas and the pixel sizes and spacings are optimised according to the JTMA as described above in Sec.~\ref{sec:pixel basis design}. The experimental setup is depicted in Fig.~\ref{fig:ExpSetup}~b). While not discussed in the previous section, an additional consideration involves the use of Gaussian beam simulations to ensure that mode waists and crystal conjugate planes are optimally located in the setup (see Appendix~\ref{sec:appendix A.I}). A grating-stabilised CW laser at 405~nm is loosely focused onto a 5~mm ppKTP crystal to produce photon pairs at 810~nm entangled in their transverse position-momentum via type-II spontaneous parametric down-conversion (SPDC). Single-outcome projective measurements in the optimized macro-pixel basis (or any of its MUBs) are performed on the photon pairs using diffractive holograms implemented on spatial light modulators (SLMs), \nh{where each macro-pixel mode is defined by displaying a grating over the corresponding circular area. To perform transformations over the modes of interest, different phases for each macro-pixel are encoded as a relative shift in the grating with respect to the other macro-pixels. As shown in Figure 2a), while the black area around the macro-pixel modes is effectively “switched off” because no grating is displayed, the grating inside each macro-pixel effectively “switches on” the selected modes so they efficiently couple to single-mode fibers. After the SLMs, }two intensity-flattening telescopes (IFT) are used for ensuring precise, generalized spatial-mode measurements across the entire modal bandwidth of interest (see Appendix~\ref{sec:appendix A.I}).

\begin{figure*}[t!]
\centering\includegraphics[width=\textwidth]{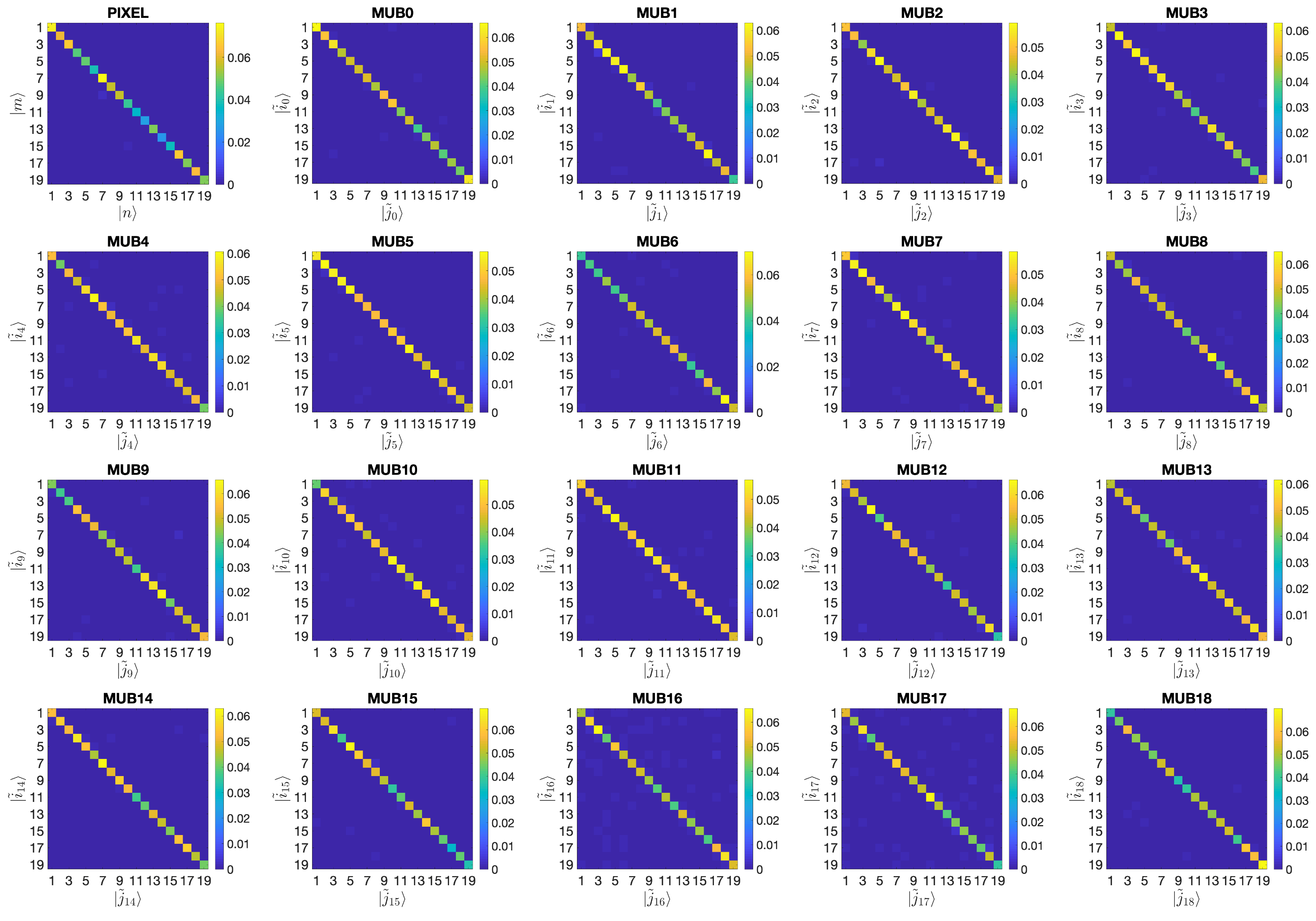}
\vspace*{-4mm}
\caption{\textbf{Experimental data in 19 dimensions.} Normalised two-photon coincidence counts showing correlations in the standard 19-dimensional pixel basis $\{\ket{mn}\}_{mn}$ and its 19 mutually unbiased bases $\{\ket{\tilde{i}_{k}\tilde{j}_{k}^{*}}\}_{i,j}$. With these measurements, we obtain a fidelity of $F(\rho,\Phi^+) = 94.4 \pm 0.4$\% to a 19-dimensional maximally entangled state, certifying 18-dimensional entanglement.}
\label{fig:19dim}
\end{figure*}

\begin{smallboxtable}[floatplacement=t]{Fidelities $F(\rho,\Phi^+)$ to the maximally entangled state, obtained via measurements in all mutually unbiased bases in dimension $d$}{Flist}
\begin{center}
\begin{tabularx}{\textwidth}{p{0.12\textwidth}p{0.12\textwidth}p{0.3\textwidth}p{0.3\textwidth}}
\hline
$d$ & $d_{\mathrm{ent}}$ & $F(\rho,\Phi^+)$ & $E_{\mathrm{oF}}$ (ebits)\\
\hline
3 & 3& $98.2 \pm 0.9$ \% &\ $1.4 \pm 0.1$ \\
5& 5& $97.5 \pm 0.5$ \%  &\ $2.2 \pm 0.1$ \\
7& 7& $96.4 \pm 0.7$ \% &\ $2.6 \pm 0.1$ \\
11& 11& $93.9 \pm 0.7$ \%  &\ $3.1 \pm 0.1$ \\
13& 13& $94.1 \pm 0.6$ \% &\ $3.3 \pm 0.1$ \\
17& 17& $94.3 \pm 0.3$ \% &\ $3.6 \pm 0.1$ \\
19& 18& $94.4 \pm 0.4$ \% &\ $3.8 \pm 0.1$ \\
\end{tabularx}
\end{center}
\end{smallboxtable}

Our modified entanglement witness allows characterizing the generated state by lower bounding its fidelity to the maximally entangled state through measurements in any two mutually unbiased bases. However, if one uses the standard basis and all of its MUBs, the fidelity bound becomes tight and one can estimate the exact fidelity to the maximally entangled state~\cite{Bavaresco:2018gw}. To showcase the effect of our improvements on the resulting pixel entanglement quality, we measure correlations in all MUBs for subspaces of prime dimension up to $d=19$. From these measurements, we obtain record fidelities to the maximally entangled state of $\approx 94\%$ or greater in all subspaces, certifying entanglement dimensions of $d_{\mathrm{ent}}=d$ for $d=3,...,17$ (see Table~\ref{table:Flist}). In $d=19$, we obtain an estimate of exact fidelity of $94.4 \pm 0.4$\%, certifying an entanglement dimensionality of $d_{\mathrm{ent}}=18$ (and just below the fidelity bound of $94.7\%$ for $d_{\mathrm{ent}}=19$). We also calculate the entanglement-of-formation ($E_{\mathrm{oF}}$), which is an entropic measure of the amount of entanglement needed to create our state~\cite{Wootters:2001tn,Coles:2017} (please see \nf{Appendix~\ref{sec:entanglement of formation}} for more details). 
 
\nh{Our $E_{\mathrm{oF}}$ reaches a value of $3.8 \pm 0.1$ ebits in $d=19$, which is already higher than previously reported values without subtracting accidental/background counts or making any assumptions on the state
~\cite{Schneeloch:2019uw,Martin2017}}
Measured correlation data in all 20 mutually unbiased bases of the $d=19$ pixel basis are shown in Fig.~\ref{fig:19dim}. The uncertainty in fidelity is calculated assuming Poisson counting statistics and error propagation via a Monte-Carlo simulation of the experiment. 

We now demonstrate the speed of our measurement technique, which is enabled by the use of phase-only pixel-basis holograms and our entanglement witness that allows the use of any two pixel MUBs for certifying entanglement. We perform measurements in two MUBs for prime dimensions ranging from $d=19$ to $97$. In $d=19$, we are able to certify a fidelity lower bound of ($93\pm 2$)\% using a total data acquisition time of 3.6 minutes (corresponding to 722 single-outcome measurements), certifying an entanglement dimensionality of $d_{\mathrm{ent}}=18\pm 1$. In comparison, prior experiments on OAM entanglement required a measurement time two orders of magnitude larger for certifying just $d_{\mathrm{ent}}=9$ in 11 dimensions~\cite{Bavaresco:2018gw}. A summary of results for dimensions 19 and higher is shown in Table~\ref{table:FboundList}. As can be seen, the entanglement quality starts dropping above $d\approx 30$, which is a result of the limited resolution of our devices and the SPDC generation bandwidth itself. Nevertheless, we are able to certify an entanglement dimensionality of $d_{\mathrm{ent}}=55\pm 1$ (See Fig.~\ref{fig:97dim}) in a 97-dimensional space and an entanglement-of-formation of $E_{\mathrm{oF}}=4.0\pm 0.1$ ebits in a 31-dimensional space, which both constitute a record-breaking entanglement dimensionality and entanglement-of-formation certified without any assumptions on the state. A very recent experiment demonstrated comparable state fidelities and $E_{\mathrm{oF}}$ values in $d=32$ using multi-path down-conversion~\cite{Hu:2020ws}. The required measurement time in the 97-dimensional pixel space was on the order of a day, while using standard basis measurements would require $\sim$14 years of measurement time for data of similar quality.

\begin{smallboxtable}[float=b!]{Fidelity bounds $\tilde{F}(\rho,\Phi^+)$ obtained via measurements in two mutually unbiased bases in dimension $d$}{FboundList}
\begin{center}
\begin{tabularx}{\textwidth}{p{0.1\textwidth}p{0.19\textwidth}p{0.24\textwidth}p{0.3\textwidth}}
\hline
$d$ & $d_{\mathrm{ent}}^{\,\dagger \dagger}$ & $F(\rho,\Phi^+)$ & $E_{\mathrm{oF}}$ (ebits)\\
\hline
19 & 18(+1)& $93 \pm 2$ \% &\ $3.6 \pm 0.2$\\
23& 22(-1)& $92 \pm 2$ \% &\ $3.6 \pm 0.2$\\
29& 27(-1)& $90 \pm 2$ \% &\ $3.8 \pm 0.2$\\
31& 29(-1)& $92 \pm 2$ \% &\ $4.0 \pm 0.1$\\
37& 32(-1)& $84 \pm 1$ \% &\ $3.4 \pm 0.1$\\
51$^\dagger$& 38(-1)& $73 \pm 1$ \% &\ $2.8 \pm 0.1$\\
97& 55($\pm 1$)& $56 \pm 1$ \% &\ $1.9 \pm 0.1$\\
\end{tabularx}
\end{center}
\footnotesize{$^\dagger$ Our two-MUB witness remains valid even for this non-prime dimension, as the two measurement bases used ($k = 0,1$) are still mutually unbiased, and their phase relationship allows the bound to hold.}\\
\footnotesize{$^{\dagger \dagger}$ 

\nf{Here we present the entanglement dimensionalities certified by the mean values of the corresponding fidelity bounds, while other Schmidt-number thresholds within the confidence interval are shown in parenthesis.}}
\end{smallboxtable}

\begin{figure*}[htbp]
\centering\includegraphics[width=0.98\textwidth]{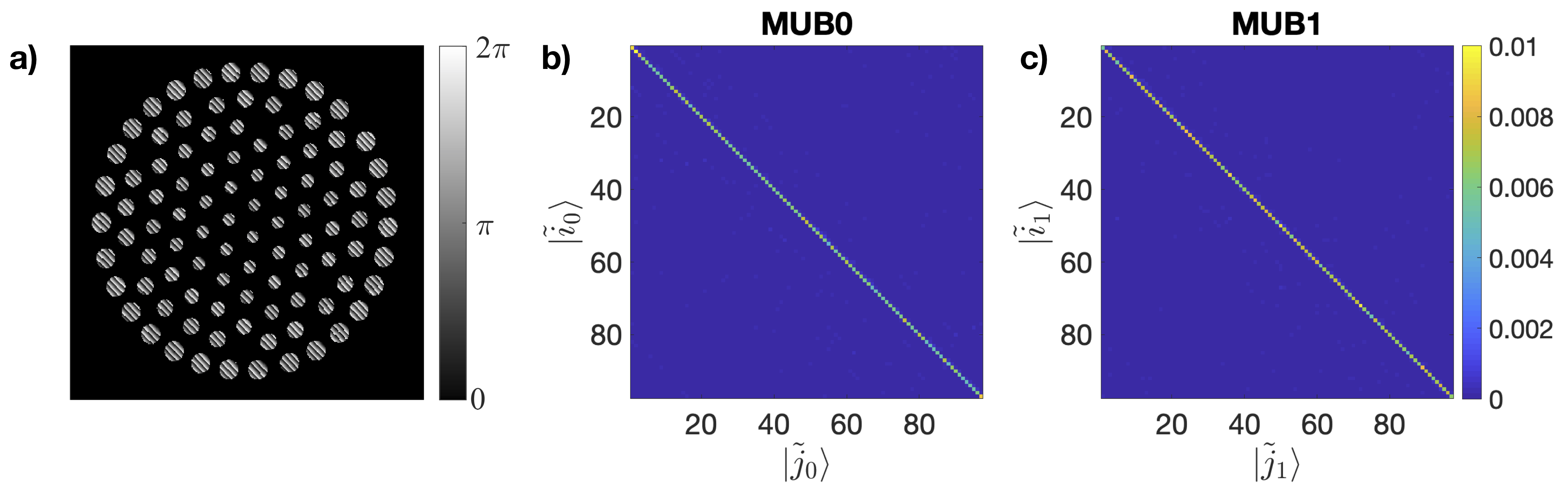}
\vspace*{-4mm}
\caption{\textbf{Experimental data in 97 dimensions.} a) An example of the diffractive hologram used for projective measurements in a 97-dimensional \nh{WF basis}. Normalised two-photon coincidence counts showing correlations in the b) first and c) second mutually unbiased bases ($k=0,1$) to the 97-dimensional pixel basis $\{\ket{\tilde{i}_{k}\tilde{j}_{k}^{*}}\}_{i,j}$. Using these measurements, we obtain a fidelity bound of $\tilde{F}(\rho,\Phi^+) = 56 \pm 1$\% that is above the bound $B_{54}=54/97=0.5567$, thus certifying an entanglement dimensionality of $d_{\mathrm{ent}}=55 \pm 1$.}
\label{fig:97dim}
\end{figure*}


\section{Conclusion}

We have demonstrated the certification of photonic high-dimensional entanglement in the transverse position-momentum degree-of-freedom with a record quality, measurement time, and entanglement dimensionality. These results are made possible through the combination of three new methods: tailored design of the spatial-mode basis, precise two-photon spatial-mode measurements, and a versatile entanglement witness based on measurements in mutually unbiased bases (MUBs). As a demonstration of the quality of our entanglement, we achieve state fidelities of $\approx 94\%$ or above in local dimensions of 19 or below, certifying up to 18-dimensional entanglement. In addition, the use of any two MUBs enables the measurement of 18-dimensional entanglement in 3.6 minutes, a reduction in measurement time by more than two orders of magnitude over previous demonstrations~\cite{Bavaresco:2018gw}. Finally, we are able to certify an entanglement dimensionality of at least 55 local dimensions and an entanglement-of-formation of 4 ebits, which to our knowledge is the highest amount certified without any assumptions on the state. While we have used projective single-outcome measurements in our experiment, recent progress on generalized multi-outcome measurement devices~\cite{Mirhosseini:2013em,Fontaine:2019ja,Brandt:2020er} and superconducting detector arrays~\cite{Wollman:2019kb} promises to increase measurement speeds even further. Our results show that high-dimensional entanglement can indeed break out of the confines of an experimental laboratory and enable noise-resistant entanglement-based quantum networks that saturate the information-carrying potential of a photon. 


\begin{acknowledgements}
  This work was made possible by financial support from the QuantERA ERA-NET Co-fund (FWF Project I3773-N36) and the UK Engineering and Physical Sciences Research Council (EPSRC) (EP/P024114/1). NF acknowledges support from the Austrian Science Fund (FWF) through the project P 31339-N27. MH acknowledges funding from the Austrian Science Fund (FWF) through the START project Y879-N27. MP acknowledges funding from VEGA project 2/0136/19 and GAMU project MUNI/G/1596/2019. 
\end{acknowledgements}


\bibliographystyle{apsrev4-1fixed_with_article_titles_full_names}
\bibliography{extrarefs}

\clearpage
\onecolumngrid
\appendix
\section*{Appendix}
\renewcommand{\thesubsection}{A.\Roman{subsection}}
\renewcommand{\thesection}{}
\setcounter{equation}{0}
\numberwithin{equation}{section}
\renewcommand{\theequation}{A.\arabic{equation}}
\setcounter{figure}{0}
\renewcommand{\thefigure}{A.\arabic{figure}}
\renewcommand{\theHfigure}{A.\arabic{figure}}

\input{Supplementary}

\end{document}

%% file: Supplementary.tex
We have demonstrated photonic high-dimensional entanglement in the discretized transverse position-momentum degree of freedom that allows the assumption-free certification of entanglement with a record quality, measurement time, and entanglement dimensionality. In this appendix we provide additional information on the experimental setup used for engineering our entangled state, give a detailed proof of the fidelity bound that we use for certifying high-dimensional entanglement, and provide a description of the entanglement of formation bound we have used in the main text for quantifying the amount of entanglement of the produced state.

\subsection{Experimental Setup}\label{sec:appendix A.I}
A nonlinear ppKTP crystal ($1~mm \times 2~mm \times 5~mm$) is pumped with a continuous wave grating-stabilized 405~nm laser in order to generate a pair of photons at 810~nm entangled in their transverse position-momentum via the process of type-II spontaneous parametric down conversion (SPDC). The crystal is temperature-tuned and phase-matching conditions are met by housing it in a custom-built oven that keeps it at 30$^{\circ}$C. With the purpose of increasing the dimensionality of the generated state, we consider Gaussian beam propagation and the transformation by two lenses to determine how to shape our beam and loosely focus it on the nonlinear crystal. As shown in Fig.~\ref{fig:pumpfocusing}, the pump laser goes through a telescope system where the focal lengths and position of the lenses are chosen such that the final beam waist is located at the crystal with a large enough size to increase the number of generated modes, but without clipping the beam by the crystal aperture. 

\begin{figure*}[b!]
\centering\includegraphics[width=0.7\textwidth]{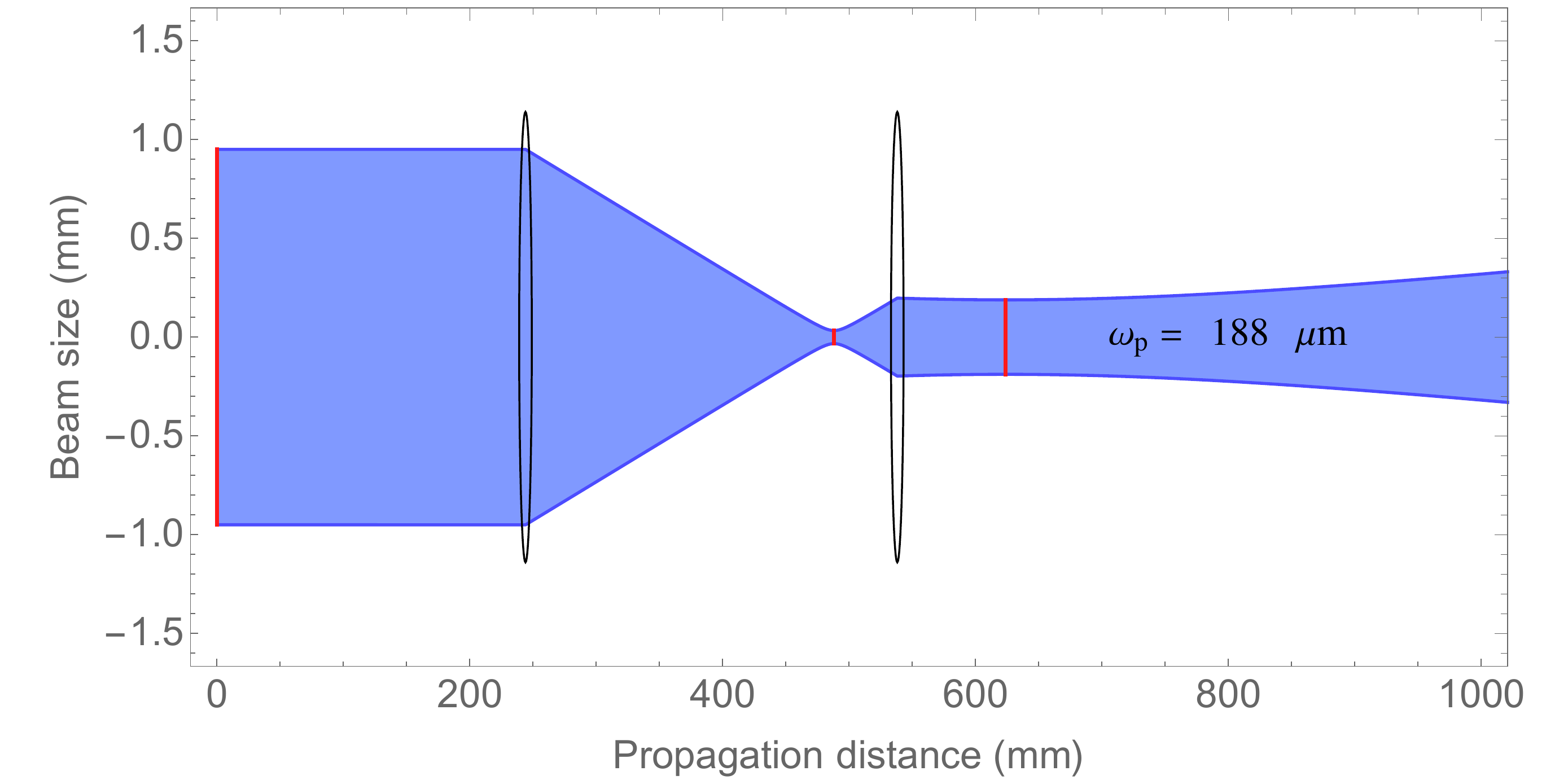}
\caption{\textbf{Gaussian beam propagation for the pump.} The collimated UV pump (beam waist $\omega_0 = 950\mu$m) propagates through a telescope system composed of two lenses with $f_1 = 250$~mm and $f_2 = 50$~mm. After the telescope, the pump beam waist is located at a distance of 50~mm away from the second lens, with a size of $w_\textrm{p} = 188\mu$m. We place the crystal such that this final beam waist is at its longitudinal center.}
\label{fig:pumpfocusing}
\end{figure*}

After the pump is removed with a dichroic mirror (DM), the entangled pair of photons is separated with a polarising beam-splitter (PBS). Each of the photons is made incident on a phase-only spatial light modulator (SLM, Hamamatsu X10468-02) that is placed in the Fourier plane of the crystal using a lens (see Fig.~\ref{fig:SPDCatSLM}). For the reflected photon to be manipulated by the SLM, we use a half-wave-plate (HWP) to rotate its polarisation from vertical to horizontal. Computer-generated holograms displayed on each SLM allow us to select particular spatial modes of the incident light and convert them into a Gaussian mode, which effectively couples (using a 10X objective) into a single-mode fiber (SMF) that carries these filtered photons to a single photon avalanche detector (SPAD). In effect, this allows us to perform projective measurements of any complex spatial mode. Time-coincident events between the two detectors are registered by a coincidence counting logic (CC). \nh{Thus, the states considered in our experiment are post-selected on detecting two photons.}

\begin{figure*}[t!]
\centering\includegraphics[width=0.8\textwidth]{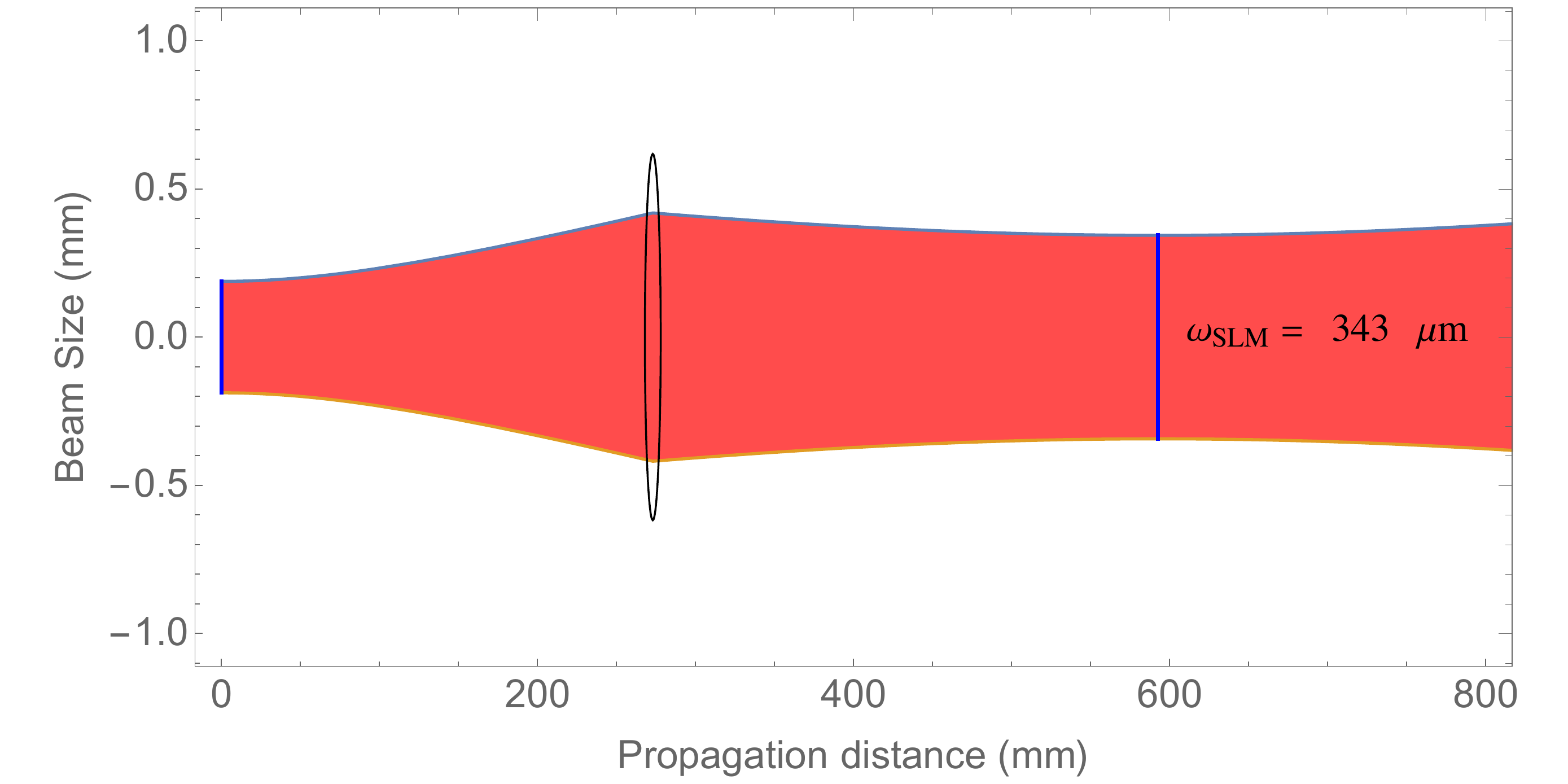}
\caption{\textbf{Gaussian beam propagation from Crystal to SLM.} For the design of the experimental setup, we model the propagation of the generated photons to the spatial light modulators assuming a Gaussian mode that has an initial beam waist of $\omega_{\textrm{SPDC}} = 188$~$\mu$m, which corresponds to the pump beam waist at the longitudinal center of the crystal. Using a lens of $f=250$~mm, we place the Fourier plane of the crystal at the SLM (represented in the figure by the blue line at the right). A Gaussian mode is this plane would have a beam waist of $\omega_{\textrm{SLM}} = 343$~$\mu$m on the SLM after propagation from the crystal.}
\label{fig:SPDCatSLM}
\end{figure*}

\begin{figure*}[b!]
\centering\includegraphics[width=0.7\textwidth]{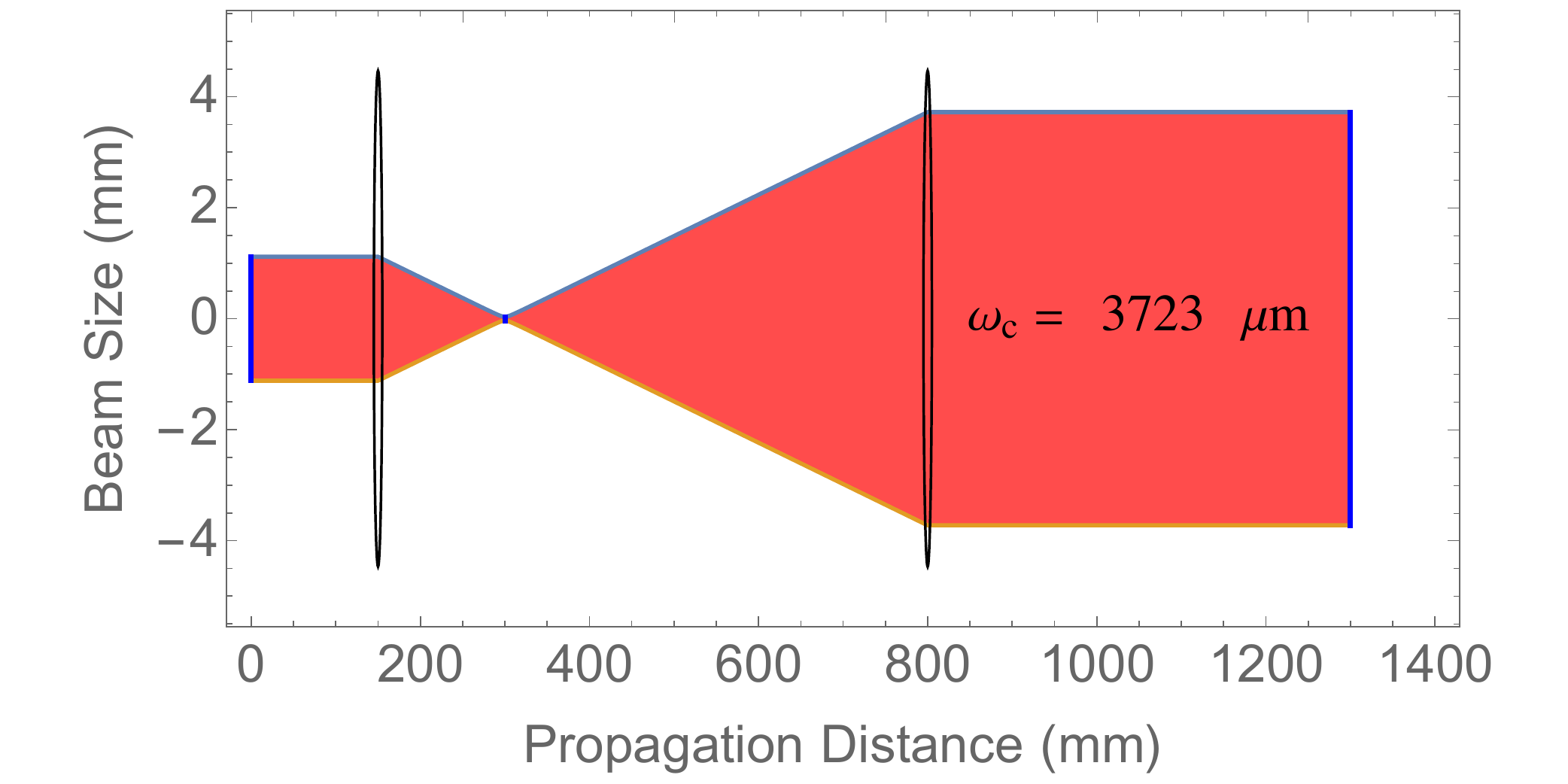}
\caption{\textbf{Intensity Flattening Telescope.} Consider a virtually backward-propagating collimated beam that is launched by a 10X microscope objective with an initial beam waist of $\omega_{\textrm{BP}} = 1.117$~mm. The intensity flattening telescope magnifies the collection mode by a factor of 3.3 using lenses with $f_1 = 150$~mm and $f_2=500$~mm. The resulting collection mode beam waist in the plane of the SLM is $\omega_\textrm{c} = 3723$~$\mu$m, which determines the maximum mode size that we can use for the pixel holograms without introducing large amounts of loss or measurement imprecision.}
\label{fig:IFTbackprop}
\end{figure*}

The accuracy of the projective measurement performed by the combination of an SLM and SMF is ensured through the use of intensity flattening telescopes (IFT)~\cite{Bouchard:2018hr}. \nh{As represented by the purple curves in Fig. 1a), the use of single-mode fibres introduces a Gaussian component into the collected JTMA and restricts the detection of higher-order modes. In order to reduce this negative effect, we use the IFT to afocally decrease the size of the mode that is propagating from the SLM to the objective lens, thus removing the Gaussian component introduced by the use of a SMF and recovering the orthogonality between spatial modes of a given basis. This technique can be better understood when one considers the overlap between the generated and collected modes on the plane of the SLM, where the telescopes magnify the back-propagated collection mode and effectively “flatten” the intensity distribution of the Gaussian component on the collected state.}
This is especially important when measuring coherent superpositions of spatial modes or pixel MUBs. In our experiment, the back-propagated mode is magnified by a factor of 3 (see Fig.~\ref{fig:IFTbackprop}), which allows us to increase the bandwidth of the measured modes while minimizing losses that would hinder our efficiency.  

The quality of our measurements also relies on the design of our bases. We choose to work with the Pixel basis, a discrete position-momentum basis composed of circular macro-pixels arranged inside a circle. The mutually unbiased bases to the pixel basis are a coherent superposition of the pixel states with a specific phase relationship between them~\cite{Wootters1989}. As described in the main text, in order to obtain the highest fidelities to a maximally entangled state, we try and equalise the probabilities of measuring each pixel mode by adapting the size of the macro-pixels, the spacing between them, and the size of the circle they are in, according to the characteristics of the joint-transverse-momentum-amplitude (JTMA) of the state we are producing. Doing so, we improve the visibility of the correlations in mutually unbiased pixel bases, which increases the fidelity of our state to the maximally entangled state.

\clearpage
\subsection{Derivation of the dimensionality witness}\label{sec:appendix A.II}

\nf{In this section, we provide a derivation of the fidelity lower bound based on two measurement bases that was originally introduced in Ref.~\cite{Bavaresco:2018gw}. In Sec.~\ref{sec:some name}, we then present an adaption of these fidelity witnesses that permits us to increase the speed with which the required measurements can be performed in our experiment.} 

Consider a bipartite quantum system with Hilbert space $\mathcal{H}\SAB=\mathcal{H}\SA\otimes\mathcal{H}\SB$, with equal local dimensions $\dim(\mathcal{H}\SA)=\dim(\mathcal{H}\SB)=d$, and an a priori unknown state $\rho$ of this system. For the certification of the Schmidt number w.r.t. the bipartition into subsystems $A$ and $B$, we want to use only a few measurement settings (local product bases) to give a lower bound on the fidelity $F(\rho,\Phi)$ to the target state $\ket{\Phi}=\sum_{n}\lambda_{n}\ket{nn}\SAB$, given by:
\begin{align}
    F(\rho,\Phi) = \tr(\ketbra{\Phi}{\Phi}\rho) = \sum_{m,n=0}^{d-1}\brakket{mm}{\rho}{nn}.
    \label{eq:Fidelityexplicit}
\end{align}
The entanglement dimensionality can be deduced from the fidelity taking into account that for any state $\rho$ of Schmidt number $r\le d$, the fidelity of Eq.~(\ref{eq:Fidelityexplicit}) is bounded by:
\begin{align}
    F(\rho,\Phi) \le B_r(\Phi):= \sum_{m=0}^{r-1}\lambda_{i_m}^2,
    \label{eq:Fidelitybounded}
\end{align}
where the sum runs over the $i_m$ with $m\in\{0,...,d-1\}$ such that $\lambda_{i_m} \ge \lambda_{i_m\pri} \forall m\ge m\pri$. Hence, any state with $F(\rho,\Phi) > B_r(\Phi)$ must have an entanglement dimensionality of at least $r+1$.

Our goal is to obtain a (lower bound on the) fidelity that is as large as possible for the target state whose Schmidt rank (the number of non-vanishing coefficients $\lambda_{n}$) is as close as possible to the local dimension $d$. Here, we want to focus on the case where the target state is the maximally entangled state $\ket{\Phi^+}$, i.e., where $\lambda_{n}=1/\sqrt{d}\,\forall\,n$. 

The method for certifying high-dimensional entanglement (in particular, the Schmidt number) described in Ref.~\cite{Bavaresco:2018gw} works in the following way. First, one designates a standard basis $\{\ket{mn}\}_{m,n=0,\ldots,d-1}$, measures locally w.r.t. this basis, obtaining estimates for the matrix elements $\{\bra{mn}\rho\ket{mn}\}_{m,n}$ from the coincidence counts $\{N_{ij}\}_{i,j}$ in the chosen basis via
\begin{align}
    \bra{mn}\rho\ket{mn}    &=\,\frac{N_{mn}}{\sum_{i,j}N_{ij}}.
    \label{eq:clicks vs matrix elements}
\end{align}
Then, one measures in $M$ of $d$ possible mutually unbiased bases $\{\ket{\tilde{i}_{k}\tilde{j}_{k}^{*}}\}_{i,j}$, where the label $k\in\{0,1,\ldots,d-1\}$ labels the chosen basis, and the asterisk denotes complex conjugation w.r.t. the standard basis. The local basis vectors are constructed according to the prescription by Wootters and Fields~\cite{Wootters1989}, that is,
\begin{align}
    \ket{\tilde{j}_{k}} &= \tfrac{1}{\sqrt{d}} \sum_{m=0}^{d-1} \omega^{jm+km^{2}} \ket{m},
    \label{eq:tiltedbasis}
\end{align}
where $\omega=\exp(\tfrac{2\pi \nr i}{d})$ are the complex $d$-th roots of unity, and we hence refer to these bases as Wootters-Fields (WF) bases. When $d$ is an odd prime, which we will assume from now on, the set of all bases $\{\ket{\tilde{j}_{k}}\}_{j}$ together with the standard basis $\{\ket{m}\}_{m}$ forms a complete set of $d+1$ mutually unbiased bases with the property that the overlaps between any two basis states of any two different bases from the set have the same magnitude. Measurements in any of these $d$ bases then provide the matrix elements  $\{\bra{\tilde{i}_{k}\tilde{j}_{k}^{*}}\rho\ket{\tilde{i}_{k}\tilde{j}_{k}^{*}}\}_{i,j}$ for any of the $M$ chosen values of $k$.

For simplicity, let us now first concentrate on the case where one measures only in a single of these Wootters-Fields bases, i.e., the case $M=1$. We use the corresponding matrix elements to bound the target state fidelity $F(\rho,\Phi)$ by splitting it in two contributions, $F(\rho,\Phi)=F_{1}(\rho,\Phi)+F_{2}(\rho,\Phi)$, where
\begin{subequations}
\begin{align}
    F_{1}(\rho,\Phi)    &:=\tfrac{1}{d}\sum_{m}\bra{mm}\rho\ket{mm},\\[1mm]
    F_2(\rho,\Phi)   &:=\tfrac{1}{d}\sum_{m\neq n} \bra{mm}\rho\ket{nn}.
    \label{eq:F2 appendix}
\end{align}
\end{subequations}
The first term $F_{1}(\rho,\Phi)$ can be calculated directly from the measurements in the standard basis, while the term $F_2(\rho,\Phi)$ can be bounded from below by the quantity $\tilde{F}_2(\rho,\Phi)\leq F_2(\rho,\Phi)$, given by
\begin{align}
    \tilde{F}_{2}    &:=\sum\limits_{j=0}^{d-1}\bra{\tilde{j}_{k}\tilde{j}_{k}^{*}}\rho\ket{\tilde{j}_{k}\tilde{j}_{k}^{*}}
        -\tfrac{1}{d}  
	-\hspace*{-4.5mm}
    \sum\limits_{\substack{m\neq m\pri\!,  m\neq n \\ n\neq n\pri\!, n\pri\neq m\pri\\}}\hspace*{-4.5mm}
    \tilde{\gamma}_{mm\pri nn\pri}\,
    \sqrt{\bra{m\pri n\pri}\rho\ket{m\pri n\pri}\bra{mn}\rho\ket{mn}},
    \label{eq:F2 bound tilted basis appendix}
\end{align}
where we have noted that $\sum_{m,n=0}^{d-1}\!\!  \bra{mn}\rho\ket{mn}=1$ by construction due to Eq.~(\ref{eq:clicks vs matrix elements}), and the prefactor $\tilde{\gamma}_{mm\pri nn\pri}$ is given by
\begin{align}
    \tilde{\gamma}_{mm\pri nn\pri} &=
    \begin{cases}0\ \ \mbox{if}\ \ (m-m\pri-n+n\pri)\!\!\!\!\mod(d) \neq0\\[1mm]
    \displaystyle
    \tfrac{1}{d}\ \ \mbox{otherwise}.
    \end{cases}
    \label{eq:tilde gamma appendix}
\end{align}
With the fidelity bound for this particular target state at hand, one can then certify a Schmidt number of $r$ whenever one finds a fidelity (bound) larger than $\tfrac{r}{d}$, see~\cite{Bavaresco:2018gw}.

To prove this, we focus on the information given by the elements of the density matrix that we obtain from measuring in the WF basis. In particular, the sum over the correlated WF matrix elements can be split into three terms, i.e.,
\begin{align}
    \sum\limits_{j=0}^{d-1}\bra{\tilde{j}_{k}\tilde{j}_{k}^{*}}\rho\ket{\tilde{j}_{k}\tilde{j}_{k}^{*}} &=:\Sigma\,=\,\Sigma_{1}\,+\,\Sigma_{2}\,+\,\Sigma_{3},
    \label{eq:sigma def}
\end{align}
which are given by
\begin{subequations}
\begin{align}
    \Sigma_{1}  &:=\,
    \tfrac{1}{d}\sum\limits_{m,n}\bra{mn}\rho\ket{mn}=\tfrac{1}{d},\\[1mm]
    \Sigma_{2}  &:=\,\tfrac{1}{d}\sum\limits_{m\neq n}\bra{mm}\rho\ket{nn},\\[1mm]
    \Sigma_{3}  &:=
    \tfrac{1}{d^{2}}\hspace*{-4mm}
    \sum\limits_{\substack{m\neq m\pri, m\neq n \\ n\neq n\pri, n\pri\neq m\pri}}
    \hspace*{-4mm}
    \operatorname{Re}\Bigl(c_{mnm\pri n\pri}
    \bra{m\pri n\pri}\rho\ket{mn}\Bigr),
\end{align}
\end{subequations}
where we have introduced the quantity
\begin{align}\label{eq:cWithStandard}
    c_{mnm\pri n\pri} &:=\sum_{j}\omega^{j(m-m\pri-n+n\pri)+k(m^{2}-m^{\prime 2}-n^{2}+n^{\prime 2})}.
\end{align}
We can then bound the real part appearing in the summands of $\Sigma_{3}$ by their modulus, i.e.,
\begin{align}
    &\operatorname{Re}\Bigl(c_{mnm\pri n\pri}
    \bra{m\pri n\pri}\rho\ket{mn}\Bigr) \,\leq\, |c_{mnm\pri n\pri}
    \bra{m\pri n\pri}\rho\ket{mn}|
    \,=\,|c_{mnm\pri n\pri}|\cdot|\!\bra{m\pri n\pri}\rho\ket{mn}\!|\,
    \label{eq:bound Re c times matrix element}
\end{align}
and use the Cauchy-Schwarz inequality and the spectral decomposition $\rho=\sum_{i}p_{i}\ket{\psi_{i}}\!\!\bra{\psi_{i}}$ such that
\begin{align}
    |\!\bra{m\pri n\pri}\rho\ket{mn}\!|
    &=|\sum\limits_{i}
    \sqrt{p_{i}}\scpr{m\pri n\pri}{\psi_{i}}
    \sqrt{p_{i}}\scpr{\psi_{i}}{mn}|\\[1mm]
    &
    \leq
    \sqrt{\sum\limits_{i}p_{i}\scpr{m\pri n\pri}{\psi_{i}}\scpr{\psi_{i}}{m\pri n\pri}}
    \sqrt{\sum\limits_{i}p_{i}\scpr{mn}{\psi_{i}}\scpr{\psi_{i}}{mn}}
    =\sqrt{\bra{m\pri n\pri}\rho\ket{m\pri n\pri}\bra{mn}\rho\ket{mn}}.\nonumber
\end{align}
Finally, one notes that for any single basis choice (labelled by $k$), one has
\begin{align}
    |c_{mnm\pri n\pri}|   &=\,|\sum_{j}\omega^{j(m-m\pri-n+n\pri)}|,
\end{align}
which vanishes whenever $(m-m\pri-n+n\pri)\!\!\mod(d)\neq0$, and is equal to $d$ otherwise, resulting in $\tilde{\gamma}_{mm\pri nn\pri}$ as in Eq.~(\ref{eq:tilde gamma appendix}).


\subsection{Changing the designation of `standard basis'}\label{sec:some name}

The fidelity bound we have derived is based on designating one of the measured bases as the `standard basis', while the other is a WF basis constructed w.r.t. the standard basis. However, the bases are mutually unbiased and because the target state is maximally entangled, we can also reverse their roles to obtain another fidelity bound that can differ in its value from the original bound, depending on the estimated matrix elements (i.e., depending on the experimental data). Moreover, if one has measured in several of the WF bases, any of them can be designated the `standard basis' and one may combine it either with the original standard basis, or with any of the other WF bases for which one has taken data.

\subsubsection{Exchanging standard basis and WF basis}\label{sec:this is not the best name of a subsection, this is just a tribute}

To see which modifications this entails, let us first spell out the specific relationship between these bases. Suppose we designate the WF basis $\{\ket{\tilde{n}_{k}}\}_{n}$ labelled by $k$ as the new `standard basis'. Then we can express the vectors of the original standard basis $\{\ket{m}\}_{m}$ as
\begin{align}
    \ket{m} &=\,\sum\limits_{n}\scpr{\tilde{n}_{k}}{m}\,\ket{\tilde{n}_{k}}\,=\,\tfrac{1}{\sqrt{d}}\sum\limits_{n}\omega^{-n\nr m-k\nr m^{2}}\,\ket{\tilde{n}_{k}}
    \,=\,\tfrac{1}{\sqrt{d}}\sum\limits_{n}\,c_{mn}\suptiny{0}{0}{(k,\mathrm{st.})}\,\ket{\tilde{n}_{k}},
\end{align}
where $c_{mn}\suptiny{0}{0}{(k,\mathrm{st.})}=\omega^{-n\nr m-k\nr m^{2}}$. Similarly, we can choose a decomposition into the conjugated basis vectors, i.e.,
\begin{align}
    \ket{m} &=\,\sum\limits_{n}\scpr{\tilde{n}_{k}^{*}}{m}\,\ket{\tilde{n}_{k}^{*}}\,=\,\tfrac{1}{\sqrt{d}}\sum\limits_{n}\,c_{mn}\suptiny{0}{0}{(k,\mathrm{st.})\nr *}\,\ket{\tilde{n}_{k}^{*}}.
\end{align}
We can then write the analogous expression to $\Sigma$ in Eq.~(\ref{eq:sigma def}) as
\begin{align}
    &\Sigma\suptiny{0}{0}{(k,\mathrm{st.})}\,=\,\sum\limits_{m}\bra{mm}\rho\ket{mm}
    \,=\,\tfrac{1}{d^{2}}\sum\limits_{\substack{ n,n\pri\\ q,q\pri}}\sum\limits_{m}
    c_{mn}\suptiny{0}{0}{(k,\mathrm{st.})\nr *}c_{mn\pri}\suptiny{0}{0}{(k,\mathrm{st.})}
    c_{mq}\suptiny{0}{0}{(k,\mathrm{st.})}c_{mq\pri}\suptiny{0}{0}{(k,\mathrm{st.})\nr *}
    \bra{\tilde{n}_{k}\tilde{n}_{k}^{\prime\nr *}}\rho\ket{\tilde{q}_{k}\tilde{q}_{k}^{\prime\nr *}},
    \label{eq:sigma def with mub}
\end{align}
which we can also split up into three terms, i.e.,
\begin{align}
    \Sigma\suptiny{0}{0}{(k,\mathrm{st.})}
    &=\Sigma\suptiny{0}{0}{(k,\mathrm{st.})}_{1}
    \,+\,\Sigma\suptiny{0}{0}{(k,\mathrm{st.})}_{2}
    \,+\,\Sigma\suptiny{0}{0}{(k,\mathrm{st.})}_{3},
\end{align}
where
\begin{subequations}
\begin{align}
    \Sigma\suptiny{0}{0}{(k,\mathrm{st.})}_{1}  &:=\,
    \tfrac{1}{d}\sum\limits_{m,n}\bra{\tilde{m}_{k}\tilde{n}_{k}^{*}}\rho\ket{\tilde{m}_{k}\tilde{n}_{k}^{*}}\,=\,\tfrac{1}{d},\\[1mm]
    \Sigma\suptiny{0}{0}{(k,\mathrm{st.})}_{2}  &:=\,\tfrac{1}{d}\sum\limits_{m\neq n}\bra{\tilde{m}_{k}\tilde{m}_{k}^{*}}\rho\ket{\tilde{n}_{k}\tilde{n}_{k}^{*}},\\[1mm]
    \Sigma\suptiny{0}{0}{(k,\mathrm{st.})}_{3}  &:=
    \tfrac{1}{d^{2}}\hspace*{-4mm}
    \sum\limits_{\substack{m\neq m\pri, m\neq n \\ n\neq n\pri, n\pri\neq m\pri}}\hspace*{-4mm}
    c_{mnm\pri n\pri}\suptiny{0}{0}{(k,\mathrm{st.})}
    \bra{\tilde{m}_{k}\pri \tilde{n}_{k}^{\prime\nr *}}\rho\ket{\tilde{m}_{k}\tilde{n}_{k}^{*}},
\end{align}
\end{subequations}
while all other contributions to the sum in Eq.~(\ref{eq:sigma def with mub}) vanish. The coefficient $c_{mnm\pri n\pri}\suptiny{0}{0}{(k,\mathrm{st.})}$  appearing in $\Sigma\suptiny{0}{0}{(k,\mathrm{st.})}_{3}$ is given by
\begin{align}\label{cKStd}
    c_{mnm\pri n\pri}\suptiny{0}{0}{(k,\mathrm{st.})} &=\,
    \sum\limits_{j}
    c_{jm\pri}\suptiny{0}{0}{(k,\mathrm{st.})\nr *}c_{jn\pri}\suptiny{0}{0}{(k,\mathrm{st.})}
    c_{jm}\suptiny{0}{0}{(k,\mathrm{st.})}c_{jn}\suptiny{0}{0}{(k,\mathrm{st.})\nr *}
    \,=\,\sum\limits_{j}\omega^{-j(m-m\pri-n+n\pri)},
\end{align}
and since $|c_{mnm\pri n\pri}\suptiny{0}{0}{(k,\mathrm{st.})}|=|c_{mnm\pri n\pri}|$, the coefficient $\tilde{\gamma}_{mm\pri nn\pri}$ in Eq.~(\ref{eq:tilde gamma appendix}) is unaffected. We can thus simply exchange the role of the standard basis and any one of the WF bases to obtain a fidelity bound of the form
\begin{align}
    F(\rho,\Phi)    &\geq\,
    \tfrac{1}{d}\sum_{m}\bra{\tilde{m}_{k}\tilde{m}_{k}^{*}}\rho\ket{\tilde{m}_{k}\tilde{m}_{k}^{*}}\,+\,
    \sum\limits_{m}\bra{mm}\rho\ket{mm}-\tfrac{1}{d}
    \nonumber\\[1mm]
    &\ \ -\hspace*{-4.5mm}
    \sum\limits_{\substack{m\neq m\pri\!,  m\neq n \\ n\neq n\pri\!, n\pri\neq m\pri\\}}\hspace*{-4.5mm}
    \tilde{\gamma}_{mm\pri nn\pri}\,
    \sqrt{
     \bra{\tilde{m}_{k}\pri \tilde{n}_{k}^{\prime\nr *}}\rho\ket{\tilde{m}_{k}\pri \tilde{n}_{k}^{\prime\nr *}}
     \bra{\tilde{m}_{k}\tilde{n}_{k}^{*}}\rho\ket{\tilde{m}_{k}\tilde{n}_{k}^{*}}
    }\,.\label{eq:fid bound WF plus standard}
\end{align}


\subsubsection{Using only two WF bases}

Now, let us consider what happens when we designate one of the WF bases (labelled by $k$) as the new `standard basis' and use a second WF basis, labelled by $k\pri$ as a mutually unbiased basis to construct our fidelity bound. As before, we express these bases w.r.t. to each other as
\begin{align}
    \ket{\tilde{m}_{k\pri}} &=\,\sum\limits_{n}\scpr{\tilde{n}_{k}}{\tilde{m}_{k\pri}}\,\ket{\tilde{n}_{k}}
    \,=\,\tfrac{1}{d}\sum\limits_{n,p}
    \omega^{p(m-n)+p^{2}(k\pri-k)}\,\ket{\tilde{n}_{k}}
    \,=\,\tfrac{1}{\sqrt{d}}\sum\limits_{n}\,c_{mn}\suptiny{0}{0}{(k,k\pri)}\,\ket{\tilde{n}_{k}},
\end{align}
where we have defined
\begin{align}\label{eq:twoWFbases}
    c_{mn}\suptiny{0}{0}{(k,k\pri)}  &:=
    \tfrac{1}{\sqrt{d}}\sum\limits_{p}
    \omega^{p(m-n)+p^{2}(k\pri-k)}.
\end{align}

In order to continue, let us recall the formula for quadratic Gauss sums for odd roots of unity.
\begin{align}\label{eq:QGaussSum}
g(a:d) = \sum_{i=0}^{d-1}\omega^{ai^2} = \left(\frac{a}{d}\right)\varepsilon_d\sqrt{d},
\end{align}
where $d$ is odd,  $\left(\frac{a}{d}\right)\in\{-1,0,1\}$ is the Jacobi symbol  and 
\begin{align}
\varepsilon_d = \begin{cases}
1 & \text{if } d \equiv 1 \mod 4\\
i & \text{if } d \equiv 3 \mod 4.
\end{cases}
\end{align}

This itself is not of the same form as~\eqref{eq:twoWFbases}, but is a main ingredient in its evaluation.

In fact~\eqref{eq:twoWFbases} is of a more general form, which is often denoted as generalized quadratic Gauss sum:
\begin{align}
    G(a,b,c) &= \sum\limits_{n=0}^{c-1} e^{2\pi i \frac{an^2+bn}{c}}.
\end{align}
In our special case, evaluation of such sums is given by the following lemma.

\begin{lemma} 
Let $c$ be odd and $\gcd(a,c) = 1$ the solution takes this form:
\begin{align}
G(a,b,c) = \varepsilon_c \sqrt{c} \left(\frac{a}{c}\right)e^{-2\pi i\frac{\psi(a)b^2}{c}},
\end{align}
where
\begin{align}
\varepsilon_c = \begin{cases}
1 & \text{if } c \equiv 1 \mod 4\\
i & \text{if } c \equiv 3 \mod 4.
\end{cases}
\end{align}
$\left(\frac{a}{c}\right) \in \{1,0,-1\}$ is the Jacobi symbol (in our case $1$ or $-1$, because $a$ and $c$ are coprime) and $\psi(a)$ is a number such that $4\psi(a)a \equiv 1 \mod c$.
\end{lemma}
\begin{proof}
Since $c$ is odd and $\gcd(a,c) = 1$, numbers $2,4$ and $a$ are all invertible modulo $c$. Therefore we can ``complete the square'' by rewriting $an^2 + bn$ as 
$a(n-h)^2 + k$, where $h = -\frac{b}{2a}$ and $k = -\frac{b^2}{4a}$ to get
\begin{align}
    G(a,b,c) &= \sum\limits_{n=0}^{c-1} e^{2\pi i\frac{a(n-h)^2 + k}{c}}
    \,=\, e^{2\pi i\frac{k}{c}}\sum\limits_{n=0}^{c-1} e^{2\pi i\frac{a(n-h)^2}{c}}
    \,=\, e^{-2\pi i\frac{\psi(a)b^2}{c}}\sum\limits_{n=0}^{c-1} e^{2\pi i\frac{an^2}{c}}
    \,=\,e^{-2\pi i\frac{\psi(a)b^2}{c}}g(a:c),
\end{align}
where the second to last equality follows from the fact that both sums iterate over the same set $\mathbb{Z}_c := \{0,\dots,c-1\}$.
\end{proof}

In our case we can set $a = k\pri-k$, $b=(m-n)$, and $c = d$, where $d$ is an odd prime and $\gcd(k\pri - k,d) = 1$ (as both $k\neq k'$ and $k,k\pri\in\{0,\dots,d-1\}$). Consequently, all conditions are met and we have
\begin{align}
    c_{mn}\suptiny{0}{0}{(k,k\pri)}  &:=
    \tfrac{1}{\sqrt{d}}\sum\limits_{p}
    \omega^{p(m-n)+p^{2}(k\pri-k)}\label{eq:GaussSumExpression}
    \,=\, \varepsilon_d \left(\frac{k\pri-k}{d}\right)\omega^{-\psi(k'-k)(m-n)^2}.
\end{align}
We also need the relation between the two conjugated bases, i.e.,
\begin{align}
    \ket{\tilde{m}_{k\pri}^{*}} &=\,\sum\limits_{n}\scpr{\tilde{n}_{k}^{*}}{\tilde{m}_{k\pri}^{*}}\,\ket{\tilde{n}_{k}^{*}}
    \,=\,\tfrac{1}{\sqrt{d}}\sum\limits_{n}\,c_{mn}\suptiny{0}{0}{(k,k\pri)\nr *}\,\ket{\tilde{n}_{k}^{*}}.
\end{align}
We can then follow the exact same steps as in Sec.~\ref{sec:this is not the best name of a subsection, this is just a tribute}: We first form the sum over correlated matrix elements in the second basis,
\begin{align}
    &\Sigma\suptiny{0}{0}{(k,k\pri)}\,=\,\sum\limits_{m}\bra{\tilde{m}_{k\pri}\tilde{m}_{k\pri}}\rho\ket{\tilde{m}_{k\pri}\tilde{m}_{k\pri}}
    \,=\,\tfrac{1}{d^{2}}\!\!\sum\limits_{\substack{ n,n\pri\\ q,q\pri}}\sum\limits_{m}
    c_{mn}\suptiny{0}{0}{(k,k\pri)\nr *}c_{mn\pri}\suptiny{0}{0}{(k,k\pri)}
    c_{mq}\suptiny{0}{0}{(k,k\pri)}c_{mq\pri}\suptiny{0}{0}{(k,k\pri)\nr *}
    \bra{\tilde{n}_{k}\tilde{n}_{k}^{\prime\nr *}}\rho\ket{\tilde{q}_{k}\tilde{q}_{k}^{\prime\nr *}},
    \label{eq:construction standard as mub}
\end{align}
which can again be split up into three terms, i.e.,
\begin{align}
    \Sigma\suptiny{0}{0}{(k,k\pri)}
    &=\Sigma\suptiny{0}{0}{(k,k\pri)}_{1}
    \,+\,\Sigma\suptiny{0}{0}{(k,k\pri)}_{2}
    \,+\,\Sigma\suptiny{0}{0}{(k,k\pri)}_{3},
\end{align}
where
\begin{subequations}
\begin{align}
    \Sigma\suptiny{0}{0}{(k,k\pri)}_{1}  &:=\,
    \tfrac{1}{d}\sum\limits_{m,n}\bra{\tilde{m}_{k}\tilde{n}_{k}^{*}}\rho\ket{\tilde{m}_{k}\tilde{n}_{k}^{*}}\,=\,\tfrac{1}{d},\\[1mm]
    \Sigma\suptiny{0}{0}{(k,k\pri)}_{2}  &:=\,\tfrac{1}{d}\sum\limits_{m\neq n}\bra{\tilde{m}_{k}\tilde{m}_{k}^{*}}\rho\ket{\tilde{n}_{k}\tilde{n}_{k}^{*}},\\[1mm]
    \Sigma\suptiny{0}{0}{(k,k\pri)}_{3}  &:=
    \tfrac{1}{d^{2}}
    \hspace*{-4mm}
    \sum\limits_{\substack{m\neq m\pri, m\neq n \\ n\neq n\pri, n\pri\neq m\pri}}\hspace*{-4mm}
    c_{mnm\pri n\pri}\suptiny{0}{0}{(k,k\pri)}
    \bra{\tilde{m}_{k}\pri \tilde{n}_{k}^{\prime\nr *}}\rho\ket{\tilde{m}_{k}\tilde{n}_{k}^{*}},
\end{align}
\end{subequations}
and one can once more confirm that all other terms vanish. The main difference now lies in the form of the coefficient $c_{mnm\pri n\pri}\suptiny{0}{0}{(k,k\pri)}$, 
which, with help of Eq.~\eqref{eq:GaussSumExpression} can be written as
\begin{align}
  c_{mnm\pri n\pri}\suptiny{0}{0}{(k,k\pri)}
  &=\,     
  \sum\limits_{j}
    c_{jm\pri}\suptiny{0}{0}{(k,k\pri)\nr *}c_{jn\pri}\suptiny{0}{0}{(k,k\pri)}
    c_{jm}\suptiny{0}{0}{(k,k\pri)}c_{jn}\suptiny{0}{0}{(k,k\pri)\nr *}
    \,=\,   (\varepsilon_d^*)^2\varepsilon_d^2\left(\frac{k\pri-k}{d}\right)^4 
    \sum\limits_{j}
  \omega^{\psi(k'-k)[(j-m\pri)^2{-(j-n\pri)^2} {-(j-m)^2} +{(j-n)^2}]}\nonumber\\
  &=\omega^{\psi(k'-k)({m^{\pri2}-m^2-n^{\pri2}+n^2})} \sum\limits_{j}
  \omega^{\psi(k'-k)2j(m-m\pri  - n+n\pri )}.
  \label{eq:cWithk2}
\end{align}
It remains to be noted that $\vert  c_{mnm\pri n\pri}\suptiny{0}{0}{(k,k\pri)} \vert =  \vert c_{mnm\pri n\pri}\vert$. We can hence write the fidelity bound in the usual form
\begin{align}
    F(\rho,\Phi)    &\geq\,
    \tfrac{1}{d}\sum_{m}\bra{\tilde{m}_{k}\tilde{m}_{k}^{*}}\rho\ket{\tilde{m}_{k}\tilde{m}_{k}^{*}}\,
    +\,
    \sum\limits_{m}\bra{\tilde{m}_{k\pri}\tilde{m}_{k\pri}^{*}}\rho\ket{\tilde{m}_{k\pri}\tilde{m}_{k\pri}^{*}}
    -\,\tfrac{1}{d}\nonumber\\[1mm]
    &\ \ -\hspace*{-4.5mm}
    \sum\limits_{\substack{m\neq m\pri\!,  m\neq n \\ n\neq n\pri\!, n\pri\neq m\pri\\}}\hspace*{-4.5mm}
    \tilde{\gamma}_{mnm\pri n\pri}\,
    \sqrt{
     \bra{\tilde{m}_{k}\pri \tilde{n}_{k}^{\prime\nr *}}\rho\ket{\tilde{m}_{k}\pri \tilde{n}_{k}^{\prime\nr *}}
     \bra{\tilde{m}_{k}\tilde{n}_{k}^{*}}\rho\ket{\tilde{m}_{k}\tilde{n}_{k}^{*}}
    }\,.\label{eq:fid bound two WF bases}
\end{align}
For any two chosen bases, the bound can hence be evaluated in a straightforward way.


\subsection{Entanglement of Formation}\label{sec:entanglement of formation}

In this section, we discuss how we bound the entanglement of formation $(E_{\mathrm{oF}})$ of our bipartite state from the measurement data for any two MUBs. Here, the $E_{\mathrm{oF}}$ quantifies how many Bell states are required to convert to a single copy of our high-dimensional entangled state via local operations and classical communication (LOCC)~\cite{Wootters:2001tn}.

Using the uncertainty relations reviewed in Ref.~\cite{Coles:2017}, the bound on the $E_{\mathrm{oF}}$ is given by
\begin{equation}
    E_{\mathrm{oF}} \geq \log_2(d) - \mathcal{H}(A_1|B_1) - \mathcal{H}(A_2|B_2), \label{eq:eoff}
\end{equation}{}
\begin{figure*}[hb!]
\centering\includegraphics[width=0.9\textwidth]{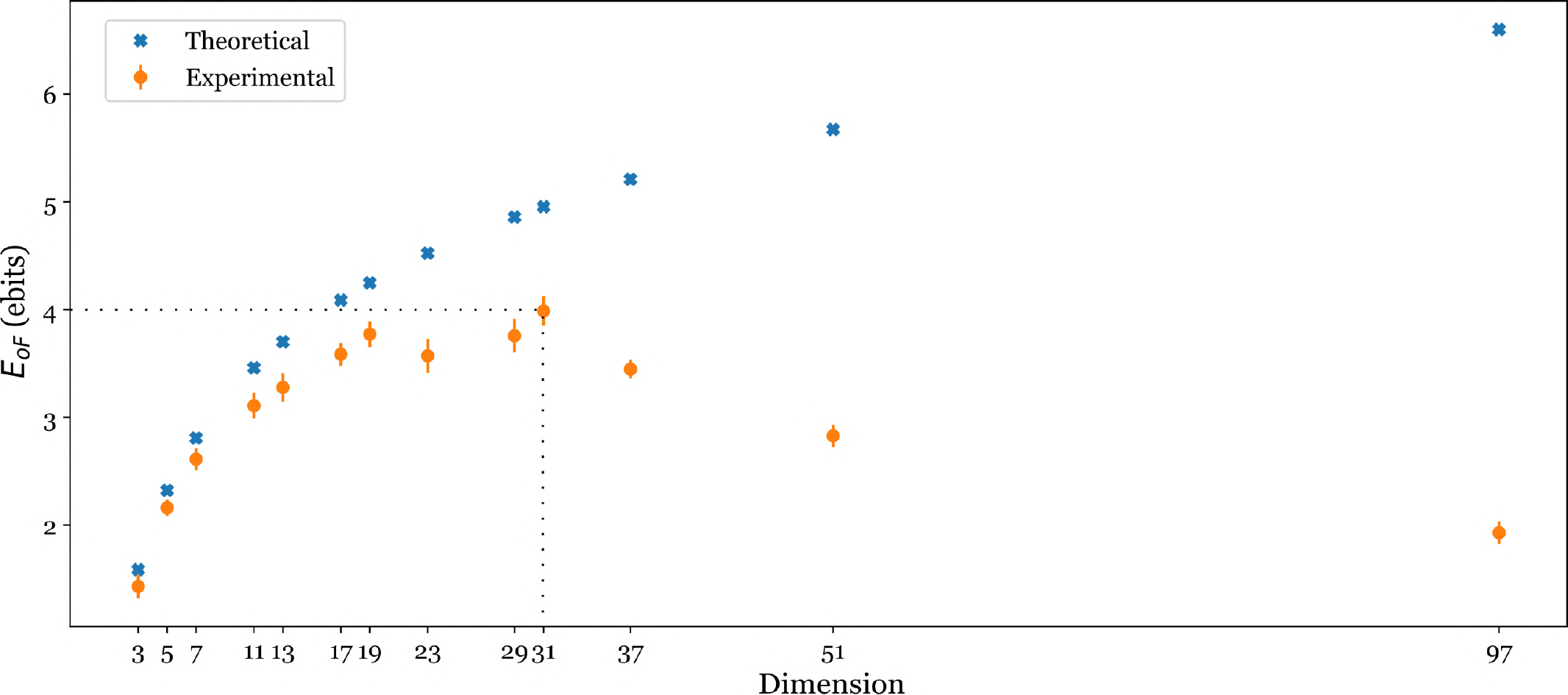}
\caption{\textbf{Entanglement of formation $(E_{\mathrm{oF}})$ bounds}. Using measurements in two mutually unbiased bases, bounds for the entanglement of formation are determined through Eq.~\eqref{eq:eoff} for dimensions up to $d=97$ (orange points). Error bars are calculated assuming Poissonian counting statistics and error propagation via a Monte-Carlo simulation of the experiment. The theoretical upper bound given by $log_2(d)$ is depicted by blue crosses. The maximum value achieved in our experiment ($E_{\mathrm{oF}} = 4.0 \pm 0.1$ ebits for $d=31$) is indicated by dotted lines.}
\label{fig:eof}
\end{figure*}
\noindent where $\mathcal{H}(A_i|B_i)$ is the conditional Shannon entropy for the $i$-th mutually unbiased basis (MUB), i.e.,
\begin{align}
    \mathcal{H}(A_i|B_i) &= \mathcal{H}(\{\rho^{(i)}_{jk}\}) - \mathcal{H}(\{\rho^{(i)}_{j}\}), \quad \text{for }  i = 1, 2, \label{eq:eof1}
\end{align}{}
with
\begin{align}
    \rho^{(i)}_{jk} &= \brakket{jk}{\rho}{jk}_i, \quad \rho^{(i)}_j  = \sum_k \brakket{jk}{\rho}{jk}_i.
\end{align}{}

\noindent Since we know the terms $\{\rho^{(i)}_{jk}\}$ are related to coincidences measured in $i$-th MUB through Eq.~\eqref{eq:clicks vs matrix elements}, the expression in~\eqref{eq:eoff} can be evaluated for any pair of MUB measurements. Figure~\ref{fig:eof} shows the $E_{\mathrm{oF}}$ achieved in our experiment as a function of dimension, compared with the theoretical maximum of $log_2(d)$. As can be seen, the $E_{\mathrm{oF}}$ increases until approximately $d=31$, where it reaches a maximum value of $4.0\pm 0.1$ ebits. As the dimension is increased further, experimental imperfections and the inherent dimensionality limits of the source bring this number back down.

